\documentclass[article,twocolumn,10pt]{IEEEtran}
\usepackage{graphicx}
\usepackage{float}
\usepackage{amsfonts}
\usepackage{amsmath}
\usepackage{amssymb}
\usepackage{cite}
\usepackage{bm}
\usepackage{mathrsfs}
\usepackage{epsfig}
\usepackage{algorithm}
\usepackage{algorithmic}
\usepackage{amsthm}
\usepackage{color}
\usepackage{graphicx}
\usepackage{caption}
\usepackage{subcaption}

\newtheorem{remark}{Remark}

\newtheorem{proposition}{Proposition}

\begin{document}

\title{Compressive Sensing-Based Detection with Multimodal Dependent Data}

\author{\authorblockA{Thakshila Wimalajeewa  \emph{Member IEEE}, and Pramod K.
Varshney, \emph{Fellow IEEE}}
}

\maketitle\thispagestyle{empty}

\begin{abstract}
 Detection  with high dimensional   multimodal  data is a challenging problem  when there are complex inter- and intra- modal  dependencies. While  several   approaches have been  proposed for dependent   data fusion (e.g., based on copula theory), their advantages  come at a high price in terms of computational complexity.  In this paper, we treat the  detection problem with compressive   sensing (CS)  where compression at each sensor  is achieved via low dimensional random projections. CS has recently  been  exploited to solve  detection problems under various assumptions on the signals of interest, however, its potential  for  dependent data fusion  has not been     explored adequately.  We exploit the capability of CS to capture statistical  properties of uncompressed data in order to compute decision statistics for detection in the compressed domain.   First, a Gaussian approximation is employed  to perform likelihood ratio (LR)  based  detection with compressed data. In this approach, inter-modal dependence  is captured via a compressed version of the covariance matrix of the concatenated (temporally and spatially) uncompressed data vector. We show that, under certain conditions, this approach   with a small number of compressed measurements per node   leads to enhanced performance compared to detection with  uncompressed  data  using  widely considered  suboptimal approaches.
 Second, we develop a nonparametric approach    where a decision statistic based on the second order statistics  of uncompressed data  is computed in the  compressed domain. The second  approach is  promising over  other related nonparametric approaches and  the first approach   when  multimodal data is  highly  correlated at the expense of  slightly increased  computational complexity.
\end{abstract}
{\bf Keywords}:  Information fusion, multi-modal data, compressive sensing, detection theory, statistical dependence, copula theory

\footnotetext[1]{This work was supported in part by ARO grant no. W911NF-14-1-0339. A part of this work  was  presented at  ICASSP, New Orleans, LA in March  2017 \cite{Wimalajeewa_icassp17}.
The authors are with the Dept. EECS, Syracuse University, Syracuse, NY. Email: \{twwewelw,varshney\}@syr.edu }

\section{Introduction}
Multimodal data represents  multiple aspects of a phenomenon of interest (PoI) observed  using different acquisition methods or different types of sensors \cite{Lahat_Proc2015}. Due to the diversity of information,  multimodal data enhances  inference  performance compared to that with unimodal data.   Multimodal data fusion has attracted much attention in different application scenarios  such as biometric score fusion \cite{Nandakumar_PAMI08,Adall_MM13}, multi-media analysis \cite{Atrey_MS10}, automatic target recognition \cite{Zhang_AES12}, and footstep detection \cite{Jin_Fusion11} to name a few.   To obtain a unified picture of the PoI to perform a given inference task, multimodal data needs to be fused in an efficient manner. This is  a challenging problem  in many applications due to  complex inter- and intra- modal dependencies and high dimensionality of data.

When the goal is to solve a detection problem in a parametric framework,  performing  likelihood ratio (LR)  based fusion is challenging  since the computation  of the  joint likelihood functions is difficult in the presence of many unknown parameters and complex inter- and intra- modal dependencies. Different techniques have been proposed to estimate the probability density functions (pdfs)  such as histograms, and kernel based methods \cite{Nandakumar_PAMI08}. In addition to LR based methods, some feature based techniques for multimodal data fusion are discussed  in  \cite{Lahat_Proc2015}. When the marginal pdfs of data of each modality are  available (or can be estimated), which can be disparate due to the heterogeneous nature of multimodal data,  copula theory has been used to model inter-modal complex dependencies  in \cite{Mercier_2007,iyengar_tsp11,ashok_tsp11,ashok_taes11,Iyengar2011,Subramanian_2011,He_tsp2015}. While there are several copula density functions  developed  in the literature, finding the best copula function that fits a given set of data is  computationally challenging. This is because different copula functions may characterize different types
of dependence behaviors  among  random variables \cite{Nelsen2006,Mari_B1}.  Finding multivariate copula density functions  with more than two modalities is another challenge since most of the existing copula functions are derived considering the bivariate case. Thus, the benefits of the use of copula theory for likelihood ratio based fusion with multimodal dependent  data come at a higher computational  price. One of the commonly used suboptimal methods   is to neglect inter-modal dependence and   compute the likelihood functions based on  the disparate marginal pdfs  of each modality;  we call this 'the product approach' in the rest of the paper.   The product  approach leads  to poor performance when   the first order  statistics of uncompressed data under two hypotheses are not  significantly different from each other and/or the  inter-modal dependence  is strong.

In this paper, we treat  the detection  problem with heterogeneous dependent data  in a compressed domain. In the proposed framework, each node compresses its time samples via low dimensional random projections as proposed in compressive sensing (CS) \cite{candes1,candes2,donoho1,Eldar_B1} and transmits  the compressed observation vector to the  fusion center.  Thus, the communication cost is greatly reduced compared to transmitting all the high dimensional observation vectors  to  the fusion center.
While CS theory  has  mostly been exploited for sparse signal reconstruction,   its potential for solving detection problems  has also been investigated  in several recent works \cite{duarte_ICASSP06,
haupt_ICASSP07,davenport_JSTSP10,Wimalajeewa_asilomar10,
Gang_globalsip14,Bhavya_cscps14,Bhavya_asilomar14, Rao_icassp2012,Cao_Info2014,Kailkhura_WCL16,Kailkhura_TSP16,Wimalajeewa_tsipn16}.
Some of the works, such as \cite{duarte_ICASSP06,haupt_ICASSP07,Gang_globalsip14,Rao_icassp2012,Cao_Info2014,Wimalajeewa_tsipn16} focused  on constructing decision statistics in the compressed domain exploiting the sparsity prior,  some other works \cite{davenport_JSTSP10,Wimalajeewa_asilomar10,Bhavya_cscps14,Bhavya_asilomar14,Kailkhura_WCL16,Kailkhura_TSP16} considered  the detection  problem when the signals  are not necessarily sparse. When the signal to be detected is known and deterministic,    a performance loss is expected in terms of the probabilities  of detection and false alarm when performing likelihood ratio based detection  in the compressed domain compared  to that with uncompressed data \cite{davenport_JSTSP10}. However, when the signal-to-noise ratio (SNR) is sufficiently large, this loss is not significant and the compressed detector is  capable of providing a  similar performance  as  the uncompressed  detector. In \cite{Kailkhura_TSP16}, the  authors have extended  the known signal detection problem with CS to the multiple sensor case considering  Gaussian measurements. While intra-signal (temporal) dependence  was  considered with Gaussian measurements,  inter-sensor (spatial)  dependence  was ignored in \cite{Kailkhura_TSP16}.  As mentioned before, with heterogeneous multimodal data, handling inter-modal dependence is one of the key issues in developing efficient fusion strategies.   To the best of authors' knowledge,    the ability of  CS in capturing the dependence properties of uncompressed data to solve  detection  problems  has  not been well investigated in the literature.

In this paper, our goal is to exploit the potential of CS to capture dependence structures  of high dimensional data focusing on detection problems. We propose a parametric as well as a nonparametric approach  for detection with compressed data. In the first approach, we treat   the detection problem  completely in the compressed domain.   With arbitrary disparate marginal pdfs for  (temporally independent) uncompressed data of each modality, we employ a   Gaussian approximation in  the compressed domain and the joint likelihood function of spatially  dependent (over modalities)  is computed based on  multivariate Gaussian pdfs. With this approach,  dependence is captured via  a compressed version of the covariance matrix of the concatenated (over all the modalities)  uncompressed data vector. We show that,  under certain conditions,  using  a small number of compressive measurements (compared to the  original signal dimension),   better  or similar  performance can be achieved  in the compressed domain compared to performing fusion  (i). using  the product approach with uncompressed data where inter-modal dependence is completely ignored  and (ii).   when widely available copula functions are used to model dependence of  highly dependent  uncompressed data.
We  further discuss  as to  how to decide when   it is beneficial  to perform  compressed detection over suboptimal detection  with uncompressed dependent data  in terms of the Bhattacharya distance measure.

In the second approach, we exploit  the potential of  CS  to capture statistical  information of uncompressed data in the compressed domain to compute a  test   statistic for detection. When uncompressed data is dependent and highly correlated \footnote{Throughout  the paper, by 'dependent and correlated', we mean that the data is dependent and has  a  non-diagonal covariance matrix. When the data is dependent but  uncorrelated, i.e., when the dependent data has a diagonal covariance matrix, we use the term 'dependent and  uncorrelated'.} in the presence of the random phenomenon being observed (alternate hypothesis), the covariance matrix of the concatenated data  vector (over modalities) is likely to  have a different structure compared to  the case where the phenomenon is absent (null hypothesis). Thus, a decision statistic can be computed based on the covariance information.   Estimation of the covariance matrix of  uncompressed data is computationally expensive when  the signal dimension is large.  Compressive covariance sensing has been discussed in \cite{Romero_SPM16} in which the covariance matrix of uncompressed data is estimated using compressed samples.  It is noted that estimation of the complete covariance matrix is not necessary to construct a reliable test statistic for detection. Covariance based test statistics have been proposed  for spectrum sensing in \cite{Zeng_C2007,Zeng_VT09}  without considering any compression. In this paper, depending on the structure of the covariance matrix of uncompressed data, efficient test statistics for detection are computed in the  compressed domain, in contrast to the work in \cite{Zeng_C2007,Zeng_VT09}.  When the difference in second order statistics under two hypotheses is more significant than that with the first order statistics,  this approach provides better performance than the first approach with some extra computational complexity. Further, under the same conditions, this approach outperforms the energy detector with compressed as well as with uncompressed data, which  is the widely  considered nonparametric detector. Moreover, in contrast  to the energy detector, the   proposed approach is robust, with respect to the  threshold setting,  against  the uncertainties of the signal parameters under the null hypothesis.

The paper is organized as follows. In Section \ref{sec_formulation}, background on the  detection  problem  with uncompressed  dependent data is discussed.  LR  based detection with compressed dependent data is considered in Section \ref{sec_likelihood}. We also discuss when it is beneficial to perform LR based detection with compressed data compared to detection using suboptimal techniques with uncompressed data considering numerical examples.  In Section \ref{sec_covariance}, we discuss how to exploit the CS measurement scheme to construct a decision statistic based on the covariance information  of uncompressed data.  In Section \ref{sec_simulation}, CS based detection performance  is investigated  with real experimental data. Section \ref{sec_conclusion} concludes the paper.

\subsection*{Notation}
The following notation and terminology are  used throughout the paper.
Scalars are denoted by lower case letters; e.g., $x$. Lower (upper)  case boldface letters are used to denote vectors (matrices); e.g., $\mathbf x$ ($\mathbf A$).   Matrix transpose is denoted by $\mathbf A^T$. The $n$-th element of the vector $\mathbf x_j$ is denoted by both $\mathbf x_j[n]$ and $x_{nj}$ while the $(m,n)$-th element of the matrix $\mathbf A$ is denoted by $\mathbf A[m,n]$. The $j$-th column vector and the $i$-th row vector of $\mathbf A$ are denoted by $\mathbf a_j$, and  $\mathbf a^i$, respectively.  The  trace operator is denoted by  $\mathrm{tr}(\cdot)$.  The $l_p$ norm of a vector $\mathbf x$ is denoted by $||\mathbf x||_p$ while the Frobenius norm of a matrix $\mathbf A$ is denoted by $||\mathbf A||_F$.  Calligraphic letters are used to denote sets; e.g., $\mathcal U$.     We use the notation $|.|$ to denote the absolute value of a scalar, and determinant of a matrix. We use $\mathbf I_N$ to denote the identity matrix of dimension $N$ (we avoid using subscript when there is no ambiguity). The vectors  of all zeros and ones  with an appropriate dimension are denoted by $\mathbf 0$ and $\mathbf 1$, respectively.
The notation $\mathbf x \sim \mathcal N(\boldsymbol\mu, \boldsymbol\Sigma)$ denotes that the random vector  $\mathbf x$ has a  multivariate Gaussian pdf with mean vector $\boldsymbol\mu$ and covariance matrix $\boldsymbol\Sigma$.

\section{Problem Formulation and Background}\label{sec_formulation}
Let there be $L$ sensor nodes in a network deployed to solve a binary hypothesis testing  problem where the  two hypotheses are denoted by $\mathcal H_1$  and $\mathcal H_0$, respectively. The observation  vector at each node is denoted by $\mathbf x_j\in \mathbb R^N$ for $j=1,\cdots,L$.   The goal is to decide as to  which hypothesis is  true based on $\mathbf x=[\mathbf x_1^T, \cdots,\mathbf x_L^T]^T$.
\subsection{Likelihood Ratio Based Detection}
Consider the detection problem in a parametric framework where the marginal pdf  of $\mathbf x_j$ is available under both  hypotheses. Let   $\mathbf x_j$ be  distributed under $\mathcal H_1$ and $\mathcal H_0$  as
\begin{eqnarray}
\mathcal H_1&:& \mathbf x_j \sim f_1(\mathbf x_j)\nonumber\\
\mathcal H_0&:& \mathbf x_j \sim f_0(\mathbf x_j),  j=1,\cdots,L\label{obs_0}
\end{eqnarray}
respectively, where $f_i(\mathbf x_j)$ denotes  the joint pdf of $\mathbf x_j$ under $\mathcal  H_i$ for $i=0,1$ and $j=1,\cdots,L$.  The optimal test which minimizes the average probability of error  based on \eqref{obs_0} is the LR test \cite{poor1} which is given by
\begin{eqnarray}
\delta = \left\{
\begin{array}{ccc}
1 ~ &~\mathrm{if} & \frac{f_1(\mathbf x)}{f_0(\mathbf x)} > \tau\\
0 ~&  \mathrm{otherwise} ~&
\end{array}\right.\label{LLR_test_uncompressed}
\end{eqnarray}
where  $\tau$ is  the  threshold. To perform the  test in \eqref{LLR_test_uncompressed}, it is required to compute the joint pdfs  $f_1(\mathbf x)$ and $f_0(\mathbf x)$. The
optimality of the LR  test is guaranteed only when the
underlying joint pdfs  are known. When $\mathbf x_1, \cdots, \mathbf x_L$  are independent under $\mathcal H_i$ for $i=0,1$, $f_i(\mathbf x)$ can be written as $
f_i(\mathbf x) = \prod_{j=1}^L f_i(\mathbf x_j)
$
for $i=0,1$. However, this assumption may not be realistic in practical applications.  For example, consider the problem of  detection of the  presence of a common random phenomenon  in a heterogenous signal processing application where there are multiple sensors of different modalities.  The data at different nodes may follow disparate marginal pdfs  due to the differences in the physics that govern each modality. The presence of the  common random phenomenon can change the  statistics of the heterogeneous data  and make the observations at different modalities dependent \cite{iyengar_tsp11}. Thus, to detect the presence of the random phenomenon in the LR framework, computation  of the joint pdf of data collected at the multiple nodes in the presence of inter-modal dependencies is required.

There are several  approaches proposed in the literature to perform LR based detection when the exact pdf of $\mathbf x$ is not available. These techniques are  commonly categorized as parametric, nonparametric, and semi-parametric approaches.
\subsection{Copula Theory}
 In a parametric framework, copulas are used to construct a valid joint distribution describing an  arbitrary and  possibly  nonlinear dependence structure \cite{Nelsen2006,Mercier_2007,iyengar_tsp11,ashok_tsp11,ashok_taes11,Iyengar2011,Subramanian_2011,He_tsp2015}.
According to copula theory,
the pdf of $\mathbf x$ under $\mathcal H_i$ can be written as \cite{Nelsen2006},
\begin{eqnarray*}
f_i(\mathbf x) = \prod_{n=1}^N \prod_{l=1}^L f_i( \mathbf x_l[n]) c_{ni}(u^i_{n1},\cdots,u_{nL}^i)
\end{eqnarray*}
for $i=0,1$ where $c_{ni}(\cdot)$ denotes the copula density function, $u_{nl}^{i}= F(\mathbf x_l[n] | \mathcal H_i)$ with $F(x|\mathcal H_i)$ denoting  the marginal cdf of $x$ under $\mathcal H_i$.
Then, the  log LR  (LLR) can be written in the following  form:
\begin{eqnarray}
\Lambda_{\mathrm{LLR}}(\mathbf x) &=& \log \frac{f_1(\mathbf x)}{f_0(\mathbf x)} = \sum_{l=1}^L \sum_{n=1}^N  \log \frac{f_1(\mathbf x_l[n])}{f_0(\mathbf x_l[n])} \nonumber \\
&+& \sum_{n=1}^N \log \frac{c_{n1} (u_{n1}^1, \cdots, u_{nL}^1 | \phi_{n1})}{c_{n0} (u_{n1}^0, \cdots, u_{nL}^0 | \phi_{n0})}\label{eq_copula_uncomp}
\end{eqnarray}
where $\phi_{n1}$ and $\phi_{n0}$ are copula parameters under $\mathcal H_1$ and $\mathcal H_0$, respectively, for $n=1,\cdots,N$. In this case, in general,  $N$ copulas where each one is  $L$-variate are selected to model  dependence. Readers may refer to \cite{Nelsen2006,Mercier_2007,iyengar_tsp11,ashok_tsp11,ashok_taes11,Iyengar2011,Subramanian_2011,He_tsp2015} to learn more about copula theory as applicable for binary hypothesis testing problems.

One of  the fundamental challenges in copula theory is to find the  copula density function that will best fit the  given data set. Further, most of the copula density functions proposed in the literature consider the bivariate case. In order to model the dependence of  multimodal data with more than two modalities, several approaches have been proposed in the literature \cite{Subramanian_2011}, which are in general computationally complex. Thus, in order to better utilize copula theory for multimodal data fusion, these challenges need to be overcome. In the following, we consider an alternate  computationally efficient approach for multimodal data fusion in which dependence among data is modeled in a low dimensional transformed domain obtained via CS.  We also discuss  the advantages/disadvantages of modeling dependence in the compressed domain via Gaussian approximation over the copula based approach with uncompressed data.

\section{Fusion  of Spatially Dependent Data in the Compressed Domain via Likelihood Ratio Test }\label{sec_likelihood}
Let $\mathbf A_j$  be specified by a set of unique sampling vectors  $\{\mathbf a^m_{j}\}_{m=1}^{M}$ with $M < N$ for $j=1, \cdots, L$.  We assume that the  $j$-th node  compresses its observations using $\mathbf A_j$ so that the compressed measurement vector is given by,
\begin{eqnarray}
\mathbf y_j = \mathbf A_j \mathbf x_j   \label{obs_1}
\end{eqnarray}
  for $j=1,\cdots,L$ where  the $m$-th element of the vector $\mathbf A_j\mathbf x_j$ is given by $\langle \mathbf a^m_{j} , \mathbf x_j\rangle$ for $m=1,\cdots, M$  where $\langle., .\rangle$ denotes the inner product.  In CS theory, the mapping $\mathbf A_j$ is usually  selected to be a random matrix.    In the rest of the paper, we make the following assumptions: (i)  $\mathbf y_j$'s are available at the fusion center without an error, (ii). $\mathbf A_j$ is an  orthoprojector so that $\mathbf A_j\mathbf A_j^T=\mathbf I$, and (iii) the elements of $\mathbf x_j$    are independent of each other for given $j$ under both hypotheses  while  there is (spatial) dependence among $\mathbf x_1,\cdots,\mathbf x_L$ under $\mathcal H_1$ (i.e., there is spatial and temporal independence under $\mathcal H_0$ and temporal independence and spatial dependence under $\mathcal H_1$).

\subsection{Likelihood Ratio  Based Fusion  With Compressed Data}
In order to perform LR  based fusion  based on (\ref{obs_1}), the computation of the joint pdf of $\{\mathbf y_1, \cdots, \mathbf y_L\}$ is required.  When the marginal pdf of each $\mathbf x_j$ is  available, the marginal pdf of each element in $\mathbf y_j$ can be computed as in the following.  The  $m$-th element of $\mathbf y_j$, $\mathbf y_j[m]$,  can be written as,
\begin{eqnarray}
\mathbf y_j[m] = \sum_{n=1}^N \mathbf A_j[m,n] \mathbf x_{j}[n]\label{y_jm}
\end{eqnarray}
for $j=1,\cdots,L$.
 Having the marginal pdfs of $\mathbf x_j[n]$ and using the independence  assumption of $\mathbf x_j[n]$ for $n=1,\cdots,N$, the  pdf of $z=  \mathbf y_j[m]$ can be found after computing the characteristic function  of $z$.
Once the marginal pdfs  of the elements in $\mathbf y_j$ for $j=1,\cdots,L$ are found, copula theory can be used in order to find the joint pdf of the compressive measurement vectors $\mathbf y_1, \cdots, \mathbf y_L$. Letting   $u_j = F_j( \mathbf y_{p}[q])$ for $j=M(p-1)+q$ where  $p=1,\cdots, L$, $q=1, \cdots, M$, the LLR  based on copula functions can be expressed  as,
\begin{eqnarray}
&~&T_{LLR}(\mathbf y)\nonumber\\
&=&\sum_{l=1}^L \sum_{k=1}^M   \log \frac{f_1( \mathbf y_{l}[k])}{f_0( \mathbf y_{l}[k])} + \log \frac{c_1 (u_1, \cdots, u_{ML} | \phi_1^*)}{c_0 (u_1, \cdots, u_{ML} | \phi_0^*)}. \label{copula_2}
\end{eqnarray}
The second term on the right hand side of   (\ref{copula_2}) requires one to find   copula density functions  of $ML$ variables which is computationally very difficult. Since we assume that the elements in $\mathbf x_j$ are independent under any given hypothesis, each element in $\mathbf y_j$ can be approximated by a Gaussian random variable (via   Lindeberg-Feller  central limit theorem assuming that the required conditions are satisfied \cite{cramer_1946,thakshilaj5}) when $N$ is sufficiently  large.

\subsection{Likelihood Ratio Based Fusion with Compressed Data  via Gaussian Approximation}
Assume  that first and  second order statistics of the concatenated data vector $\mathbf x = [\mathbf x_1^T, \dots, \mathbf x_L^T]^T$  are available.
We define additional notation here. Let \begin{eqnarray}
\boldsymbol\beta^i= [{\boldsymbol\beta_1^i}^T \cdots  {\boldsymbol\beta_L^i}^T]^T \label{uncompressed_mean}
 \end{eqnarray} and
  \begin{eqnarray}
 \mathbf D^i =\left[
 \begin{array}{cccc}
 \mathbf D^i_1 & \mathbf D^i_{12} & \cdots & \mathbf D^i_{1L} \\
 \mathbf D^i_{21} &  \mathbf D^i_2 & \cdots & \mathbf D^i_{2L}\\
 \cdots & \cdots &  \cdots & \cdots\\
  \mathbf D^i_{L1} &  \mathbf D^i_{L2} & \cdots & \mathbf D^i_{L}.
 \end{array}\right]\label{uncompressed_cov}
 \end{eqnarray}
  denote the $NL\times 1$ mean vector  and the $NL\times NL$ covariance matrix of   $\mathbf x$ under $\mathcal H_i$ for $i=0,1$   where
     $\boldsymbol\beta_j^i = \mathbb E\{\mathbf x_j | \mathcal H_i\}$, $\mathbf D^i_j = \mathbb E\{(\mathbf x_j - \boldsymbol\beta_j^i)(\mathbf x_j - \boldsymbol\beta_j^i)^T | \mathcal H_i\}$ and $\mathbf D^i_{jk} = \mathbb E\{(\mathbf x_j - \boldsymbol\beta_j^i)(\mathbf x_k - \boldsymbol\beta_k^i)^T | \mathcal H_i\}$  for   $j\neq k$ $k=1,\cdots,L$ and $j=1,\cdots,L$.

      With Gaussian approximation,    the joint pdf of $\mathbf y = [\mathbf y_1^T \cdots \mathbf y_L^T]^T$  is given by  $\mathbf y|\mathcal H_i \sim \mathcal N (\boldsymbol\mu^i, \mathbf C^i)$ where
 $\boldsymbol\mu^i$ and $ \mathbf C^i$ are the notations used to define the  mean vector and the covariance matrix of $\mathbf y$ which are analogous  to the definitions in \eqref{uncompressed_mean} and \eqref{uncompressed_cov}, respectively, with  $\boldsymbol\mu_j^i = \mathbb E\{\mathbf y_j | \mathcal H_i\}$,  $\mathbf C^i_j = \mathbb E\{(\mathbf y_j - \mathbb E\{\mathbf y_j\})(\mathbf y_j - \mathbb E\{\mathbf y_j\})^T | \mathcal H_i\}$, $\mathbf C^i_{jk} = \mathbb E\{(\mathbf y_j - \mathbb E\{\mathbf y_j\})(\mathbf y_k - \mathbb E\{\mathbf y_k\})^T | \mathcal H_i\}$ with  $j\neq k$, $k=1,\cdots,L$ and $j=1,\cdots,L$ for $i=0,1$.   We further denote by  $\mathbf D_x$ ($\mathbf C_y$) the covariance matrix of $\mathbf x$ ($\mathbf y$) where   $\mathbf D_x=\mathbf D^1$ ($\mathbf C_y=\mathbf C^1$ ) under $\mathcal H_1$ and $\mathbf D_x=\mathbf D^0$ ($\mathbf C_y=\mathbf C^0$ ) under $\mathcal H_0$. First and second order statistics of the compressed data are related to that of uncompressed data via
 \begin{eqnarray*}
 \boldsymbol\mu_j^i = \mathbf A_j \boldsymbol\beta_j^i,
 \mathbf C^i_j = \mathbf A_j \mathbf D^i_j \mathbf A_j^T, \mathrm{and~}
 \mathbf C^i_{jk} = \mathbf A_j \mathbf D^i_{jk} \mathbf A_k^T
 \end{eqnarray*}
for $j,k=1,\cdots, L$ and $i=0,1$. Then, we can write,
\begin{eqnarray*}
\boldsymbol\mu^i = \mathbf A \boldsymbol \beta^i \mathrm{~and~}
\mathbf C^i = \mathbf A \mathbf D^i \mathbf A^T
\end{eqnarray*}
where
\begin{eqnarray}
\mathbf A = \left(
\begin{array}{ccccc}
\mathbf A_1 & \mathbf 0 & \cdot  & \cdot &\mathbf 0\\
\mathbf 0 & \mathbf A_2 & \cdot  & \cdot &\mathbf 0\\
 \cdot  &  \cdot  & \cdot  & \cdot & \cdot \\
 \mathbf 0 & \mathbf 0 & \cdot  & \cdot & \mathbf A_L\\
\end{array}\right)\label{eq_A}
\end{eqnarray}
is a $ML \times NL$ matrix. With the assumption  that $\mathbf A_j \mathbf A_j^T = \mathbf I_M$ for $j=1,\cdots,L$, the decision statistic of the LLR based detector is   given by \cite{poor1},
\begin{eqnarray}
\Lambda_{\mathrm{LLR}} (\mathbf y) &=& \frac{1}{2}\mathbf y^T ({\mathbf C^0}^{-1} - {\mathbf C^1}^{-1})\mathbf y \nonumber\\
&+&  ({\boldsymbol\mu^1}^T {\mathbf C^1}^{-1} - {\boldsymbol\mu^0}^T {\mathbf C^0}^{-1}) \mathbf y + \tau_0  \label{stat_Gaussian}
\end{eqnarray}
where $\tau_0 = \frac{1}{2}\left(\log \left(\frac{|\mathbf C^0|}{|\mathbf C^1|}\right) + {\boldsymbol\mu^0}^T {\mathbf C^0}^{-1}\boldsymbol\mu^0 -  {\boldsymbol\mu^1}^T {\mathbf C^1}^{-1}\boldsymbol\mu^1 \right)$.
To compute   the threshold so that the probability of false alarm is kept under a desired value, computation of the pdf of $\Lambda_{\mathrm{LLR}}$ under $\mathcal H_0$ is required. This  is in general computationally intractable, but is possible under certain assumptions on $\mathbf x$. For example, when $\boldsymbol\beta^i=\mathbf 0$ for  $i=0,1$ and the elements of $\mathbf x$  are identical under $\mathcal H_0$ (in addition to independence), we have $\boldsymbol\mu^i=\mathbf 0$ and  $\mathbf C^0= \sigma_0^2 \mathbf I$ where $\sigma_0^2$ denotes the variance of $\mathbf x$ under $\mathcal H_0$. In this case, the threshold can be computed as considered in \cite{poor1} (pages 73-75).  When  such assumptions on $\mathbf x$ cannot be made, we propose  to compute the threshold via simulations.

\subsubsection{Impact of compression on inter-modal dependence}\label{sec_impact_dependence}
With the Gaussian approximation after  compression, the inter-modal dependence is captured only through the covariance matrix. Higher order dependencies of data are not taken into account in the compressed domain.  In particular, $\mathbf D_x$  is compressed via  $\mathbf C_y = \mathbf A \mathbf D_x \mathbf A^T$. To quantify the distortion of $\mathbf D_x$ due to compression, one measure is to consider  the Frobenius norm of   the covariance matrix. We have
\begin{eqnarray}
||\mathbf C_y||_F^2 &=& ||\mathbf A\mathbf D_x\mathbf A^T||_F^2= \mathrm{tr}(\mathbf A \mathbf D_x^T \mathbf A^T \mathbf A \mathbf D_x \mathbf A^T )\nonumber\\
&=&\mathrm{tr}(\mathbf A^T \mathbf A \mathbf D_x^T \mathbf A^T \mathbf A \mathbf D_x  )\approx \frac{M^2}{N^2}||\mathbf D_x||_F^2
\end{eqnarray}where the last approximation is due to  $\mathbf A^T \mathbf A \approx \frac{M}{N} \mathbf I$.
Thus,  the  Frobenius norm of the covariance matrix after compression is reduced  by a factor of $c_r= \frac{M}{N}$ compared to that with uncompressed data. In other  words,  the Gaussian approximation in the compressed domain can capture a compressed version of the covariance matrix of uncompressed data. Compared to other approaches with uncompressed data, the product approach does not capture any form of dependence.  With respect to copula based approaches, it is not very clear how much dependence  can be captured with a given copula function.  Since the covariance matrix is not a direct measure of detection performance, in the following subsection we compare the detection performance of different approaches in terms of the average probability of error.

\subsection{Detection Performance Comparison Between Compressed and  Uncompressed Data via  Average Probability of Error}\label{sec_Bhatt}
In order to quantify the detection performance  of  different approaches with  both uncompressed and compressed data, we consider the Bhattacharya bound (which is a special case of the \emph{Chernoff} bound) which bounds the average probability of error of LR based detectors. We use the notation 'u:product', and 'u:copula-name' for the product approach and the copula based approach with a given copula function stated under 'name', respectively,   with  uncompressed data. The notation 'c:GA' is used to represent the LR based approach with compressed data using the Gaussian approximation.

The Bhattacharya distance between the  two hypotheses with the copula based approach with uncompressed data is given by,
\begin{eqnarray}
\mathcal D_{B}^{\mathrm{u:copula}} (f_1 || f_0) &=&  -\log \int f_1^{1/2}(\mathbf x)f_0^{1/2}(\mathbf x) d \mathbf x\nonumber\\
 &=& -\log \mathbb E_{f_0} \left\{\prod_{n=1}^N\prod_{l=1}^L\left( \frac{f_1^m(\mathbf x_l[n])}{f_0^m(\mathbf x_l[n])} \right)^{1/2} \right.\nonumber\\
&~& \left.c_{n1}^{1/2} (u_{n1}^1, \cdots, u_{nL}^1 | \phi_{n1})\right\} \label{DBt_1}
\end{eqnarray}
where  $f_i^m$ denotes the marginal pdf under $\mathcal H_i$ and  we have $f_0(\mathbf x) = \underset{l,n}{\prod} f_0^m(\mathbf x_l[n])$ since we assume    $\mathbf x_1, \cdots, \mathbf x_L$    to be independent of each other  under $\mathcal H_0$. With the product approach, we have $c_{n1}(\cdot)=1$ and \eqref{DBt_1} reduces to,
\begin{eqnarray}
&~&\mathcal D_{B}^{\mathrm{u:product}} ( (f_1 || f_0)\nonumber\\
 &=& -\log \mathbb E_{f_0} \left\{\prod_{n=1}^N\prod_{l=1}^L\left( \frac{f_1^m(\mathbf x_l[n])}{f_0^m(\mathbf x_l[n])} \right)^{1/2} \right\}. \label{DBt_product}
\end{eqnarray}
On the other hand, the Bhattacharya distance between the two hypotheses with compressed data under Gaussian approximation can be computed as \cite{Moustafa_12}
\begin{eqnarray}
\mathcal D_{B}^{\mathrm{c:GA}} (f_1 || f_0)& =& \frac{1}{8} (\boldsymbol\beta_1 - \boldsymbol\beta_0)^T \Gamma^{\dagger} \boldsymbol\beta_1 - \boldsymbol\beta_0) \nonumber\\
&+& \frac{1}{2} \log \{|\Gamma| |\mathbf A\mathbf D^1\mathbf A^T|^{-1/2}|\mathbf A\mathbf D^0\mathbf A^T|^{-1/2} \}\label{BD_Gaussian}
\end{eqnarray}
where $\Gamma^{\dagger} = \mathbf A^T \Gamma^{-1} \mathbf A$ and $\Gamma = \frac{1}{2}(\mathbf A\mathbf D^1\mathbf A^T+\mathbf A\mathbf D^0\mathbf A^T)$. Using  the Bhattacharya distance, the average probability of error with compressed data, $P_e^c$,  is upper bounded by \cite{poor1},
\begin{eqnarray*}
P_e^c \leq  \frac{1}{2} e^{-\mathcal D_{B}^{\mathrm{c:GA}}} \triangleq  P_{ub}^{\mathrm{c:GA}}.
\end{eqnarray*}

Let $\mathcal D_{B}^{\mathrm{u: gvn}}$, and  $P_{ub}^{\mathrm{u:gvn}}$    be the Bhattacharya distance, and  the upper bound namely the Bhattacharya bound  on the probability of error, respectively,  with  uncompressed data computed using a  given suboptimal approach (e.g., product or copula with a given copula function). Then, we have
\begin{eqnarray}
  P_{ub}^{\mathrm{c:GA}} &\leq& P_{ub}^{\mathrm{u:gvn}} ~~\mathrm{if} ~~\mathcal \mathcal \mathcal D_{B}^{\mathrm{c:GA}} \geq  \mathcal D_{B}^{\mathrm{u:gvn}} \label{Pe_bounds}
    \end{eqnarray}
where $\mathcal D_{B}^{\mathrm{u:gvn}}$ and $\mathcal D_{B}^{\mathrm{c:GA}}$ are computed as in \eqref{DBt_1} and \eqref{BD_Gaussian}, respectively.
In the case where uncompressed data is dependent and  uncorrelated, (\ref{Pe_bounds}) can be further simplified as stated in Proposition \ref{prop1}.
\begin{proposition}\label{prop1}
Let uncompressed data be  dependent  and  uncorrelated so that $\mathbf D^1$ is diagonal. Further, let $\mathbf D_j^i = \sigma_{j,i}^2 \mathbf I$ and $\boldsymbol\beta_j^i=\beta_{j,i}\mathbf 1$ for $i=0,1$ and $j=1,\cdots,L$. Then, we have
  \begin{eqnarray}
  P_{ub}^{\mathrm{c:GA}} &\leq& P_{ub}^{\mathrm{u:gvn}} ~~\mathrm{if} ~~\mathcal D_{B}^{\mathrm{u:gvn}} \leq c_r\rho_B ~\mathrm{and}\nonumber\\
    \end{eqnarray}
    where
    \begin{eqnarray}
\rho_B &=&\frac{N}{2}\left\{\sum_{j=1}^L  \log(\sigma_{j,1}^2 + \sigma_{j,0}^2) - \log (\sigma_{j,1}^2 \sigma_{j,0}^2) \right.\nonumber\\
&+&\left. \frac{(\beta_{j,1} - \beta_{j,0})^2}{2(\sigma_{j,1}^2 + \sigma_{j,0}^2)} \right\}\label{rho_B}
\end{eqnarray}
which   is determined by the statistics of the uncompressed data and $c_r=\frac{M}{N}$ is the compression ratio.
\end{proposition}
\begin{proof}
The proof follows from the fact that when uncompressed data is uncorrelated under $\mathcal H_1$,   $\mathcal D_{B}^{c,G} $ in \eqref{BD_Gaussian} reduces to,
\begin{eqnarray}
\mathcal D_{B}^{c,G}(f_0||f_1) = \frac{M}{N} \rho_B
\end{eqnarray}
 where $\rho_B$,  are as defined in \eqref{rho_B}.
\end{proof}
Thus, whenever $\mathcal D_{B}^{\mathrm{c:GA}} > \mathcal D_{B}^{\mathrm{u:gvn}}$ `c:GA' performs better than  any given suboptimal approach with uncompressed data. Even though $\mathcal D_{B}^{\mathrm{c:GA}} < \mathcal D_{B}^{\mathrm{u:gvn}}$,  `c:GA' can still  be promising   if the desired performance level in terms of the upper bound on the probability of error is reached. Let $\epsilon_B$ be  the desired upper bound on the probability of error. When $\mathcal D_{B}^{\mathrm{c:GA}} \geq -\log(2\epsilon_B)$,  `c:GA'  provides  the desired performance even if  $\mathcal D_{B}^{\mathrm{c:GA}} < \mathcal D_{B}^{\mathrm{u:gvn}}$.

\subsection{Illustrative Examples}\label{example}
In  the following, we consider   example scenarios  to illustrate the  detection performance with `c:GA' compared to that with uncompressed data using  different suboptimal approaches. In the two examples, two types of detection problems are  considered. In the first example, we consider a problem of detection of changes in statistics of  data collected at heterogeneous sensors. In the second example, detection of a random source by heterogeneous sensors in the presence of noise is considered.
\subsubsection{Example  1}
In the first example, we consider  $L=3$ and  a common random phenomenon
causes a change in the statistics of heterogeneous data at the three sensors. The uncompressed data at the three nodes have the following marginal pdfs: \cite{iyengar_tsp11}:
\begin{eqnarray}
x_{n1} | \mathcal H_i \sim \mathcal N(0,\sigma_i^2), ~ x_{n2} | \mathcal H_i \sim \mathrm{Exp} (\lambda_i)  \nonumber\\
 ~\mathrm{and} ~x_{n3} | \mathcal H_i \sim \mathrm{Beta}(a_i, b_i=1)
 \end{eqnarray} for $i=0,1$.
 It is noted that   $x\sim \mathrm{Exp} (\lambda)$ denotes that $x$ has an  exponential distribution with $f(x) = \lambda e^{-\lambda x}$ for $x\geq 0$ and $0$ otherwise, and $x\sim \mathrm{Beta}(a, b)$ denotes that $x$ has a beta distribution with pdf $f(x) = \frac{1}{\mathcal B(a,b)}x^{a-1} (1-x)^{b-1}$ and  $\mathcal B(a,b) = \frac{\Gamma(a)\Gamma(b)}{\Gamma(a+b)}$ is the beta function.  The data under  $\mathcal H_1$ is assumed to be dependent and the following operations are used to generate dependent data. For the data at the second node, we use
$
x_{n2} =x_{n1}^2+w^2
$ for $n=1,\cdots, N$
where $w\sim \mathcal N(0, \sigma_1^2)$. Then,  we have $x_{n2}\sim \mathrm{Exp} (\lambda_1)$ with $\lambda_1=\frac{1}{2\sigma_1^2}$. For the third node, the data under $\mathcal H_1$ is  generated as \begin{eqnarray*}
x_{n3} =\frac{u}{u+x_{n2}}
\end{eqnarray*}
for $n=1,\cdots, N$
where $u\sim \mathrm{Gamma}(\alpha_1, \beta_1=1/\lambda_1)$.  Then $x_{n3} | \mathcal H_1 \sim \mathrm{Beta}(a_1, b_1=1)$ with $a_1=\alpha_1$. It is noted that $x\sim \mathrm{Gamma}(\alpha, \beta)$ denotes that $x$ has Gamma pdf with $f(x) = \frac{1}{\beta^{\alpha}\Gamma(\alpha)} x^{\alpha-1} e^{-x / \beta}$ for $x\geq 0$ and $\alpha, \beta > 0$. Under $\mathcal H_0$, $x_{n1}$,  $x_{n2}$ and $x_{n3}$ are generated independently using the assumed marginal pdfs. In  this example, we consider three cases.

\subsubsection*{Case I}
In Case I, the data at the first and second sensors are fused.
In this case, the covariance matrices  of $\mathbf x = [\mathbf x_1^T~\mathbf x_2^T]^T$ under the  two hypotheses, $\mathbf D^0$ and $\mathbf D^1$,  are  composed of $\mathbf D_1^0 = \sigma_0^2 \mathbf I$, $\mathbf D_{12}^0=\mathbf D_{21}^0=\mathbf 0$, $\mathbf D_2^0 =  \frac{1}{\lambda_0^2}\mathbf I$  under $\mathcal H_0$ and $\mathbf D_1^1 = \sigma_1^2 \mathbf I$, $\mathbf D_{12}^1=\mathbf D_{21}^1=\mathbf 0$, $\mathbf D_2^1 =  \frac{1}{\lambda_1^2}\mathbf I$  under $\mathcal H_1$, respectively. It is  worth noting  that $\mathbf D^1$ is diagonal in this case. Thus, although $\mathbf x_1$ are $\mathbf x_2$ are spatially dependent under $\mathcal H_1$ (by construction), they are uncorrelated, , i. e., higher-order statistics exhibit dependence while the second-order correlation is zero.

\subsubsection*{Case II }\label{example2}
For Case II, we consider the fusion of data at the  second and third sensors where $\mathbf x=[\mathbf x_2^T ~ \mathbf x_3^T]^T$.
In this case, we have $\mathbf D_2^0 = \frac{1}{\lambda_0^2} \mathbf I$, $\mathbf D_{23}^0=\mathbf D_{32}^0=\mathbf 0$, $\mathbf D_3^0 =  \frac{a_0}{(a_0+1)^2(a_0+2)}\mathbf I$  under $\mathcal H_0$ and $\mathbf D_2^1 = \frac{1}{\lambda_1^2}\mathbf I$, $\mathbf D_{23}^1=\mathbf D_{32}^1=\left(\mathbb E_{x_{n1}u}\{\frac{x_{n1}u}{u+x_{n1}}\} - \frac{a_1}{\lambda_1(a_1+1)}\right)\mathbf I$, $\mathbf D_3^1 =  \frac{a_1}{(a_1+1)^2(a_1+2)}\mathbf I$  under $\mathcal H_1$, respectively. It is  noted that $\mathbf D^1$ is not diagonal in this case.

\subsubsection*{Case  III }\label{example3}
In Case III, we consider the fusion of data at all three senors.   Also $\mathbf D^1$ is not diagonal in this case as well.
\begin{figure}[h!]
\centerline{\epsfig{figure=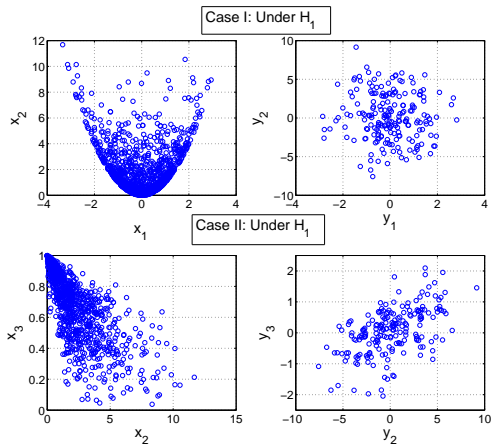,width=7.50cm}}
\caption{Scatter plots  of uncompressed and compressed data under $\mathcal H_1$ in Example 1; $N=1000$, $M=200$, $L=2$}\label{fig:scatter_L2}
\end{figure}
\subsubsection*{Scatter plots of uncompressed and compressed data}
First,  we illustrate  how the dependence  structure of   data  changes  going  from the uncompressed domain to the compressed domain. In  Fig.  \ref{fig:scatter_L2},  we show the scatter plots for both compressed and uncompressed data at the  two sensors under $\mathcal H_1$ for Cases I and II.  In Fig.  \ref{fig:scatter_L2}, the top and bottom subplots are for Case  I and Case  II, respectively while the left and right subplots are for uncompressed and compressed data, respectively. It can be observed that while uncompressed data at the two sensors are strongly dependent,  compressed data appears to be  weakly dependent. This  change of the dependence structure due to compression  was addressed in Section \ref{sec_impact_dependence}.  In this example, the scatter plots of compressed data look
more circular (Case  I) or elliptical (Case  II). In Case 1, even though  $x_{n1}$ and $x_{n2}$ for given $n$ are dependent under $\mathcal H_1$, they are uncorrelated. This leads to a circular and independent  scatter plot for compressed data for Case I.

\begin{figure}[h!]
    \centering
    \begin{subfigure}[b]{0.45\textwidth}
        \includegraphics[width=\textwidth]{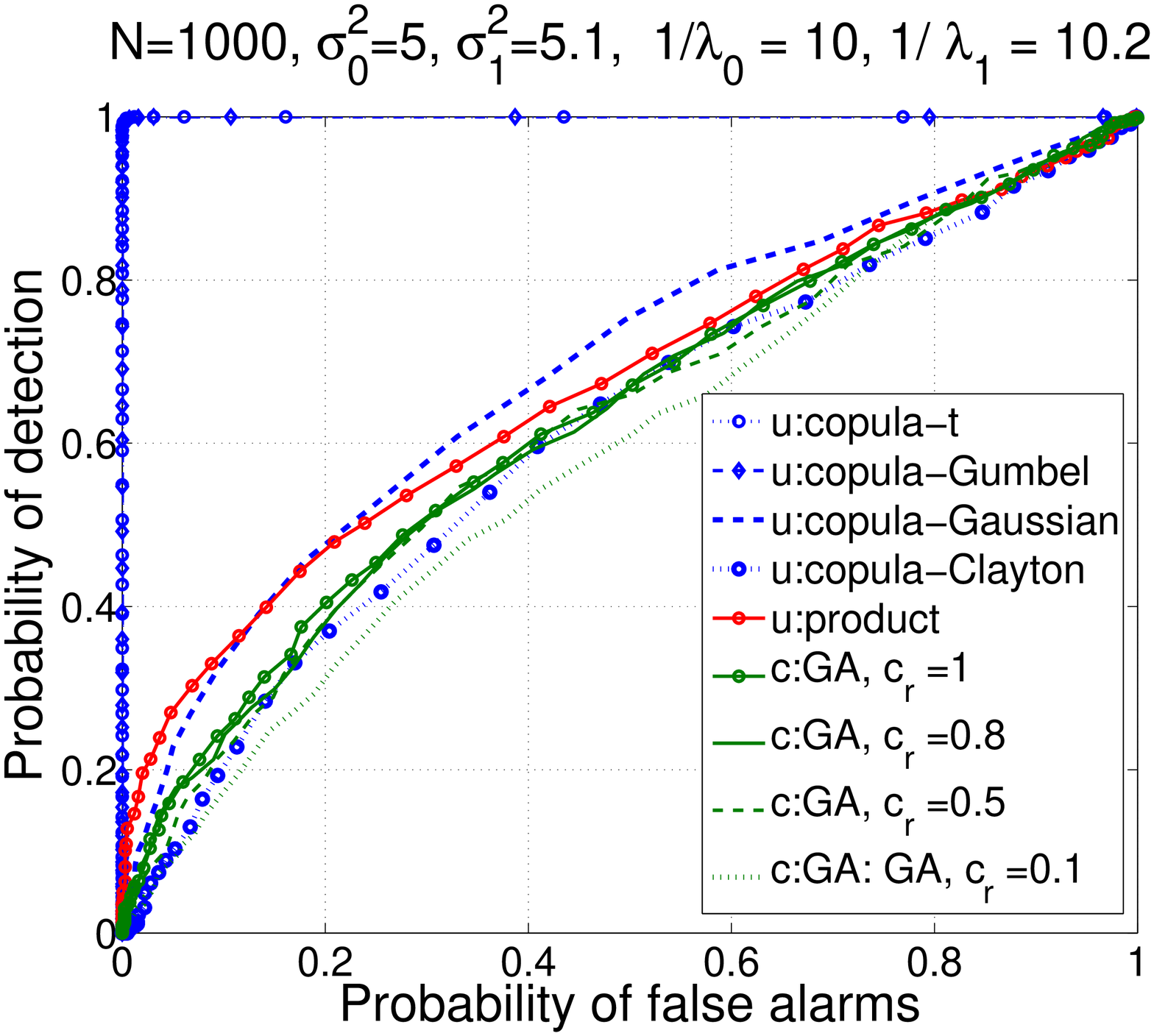}
        \caption{ Case I}
        \label{fig:N100}
    \end{subfigure}
    ~ 
    \begin{subfigure}[b]{0.45\textwidth}
        \includegraphics[width=\textwidth]{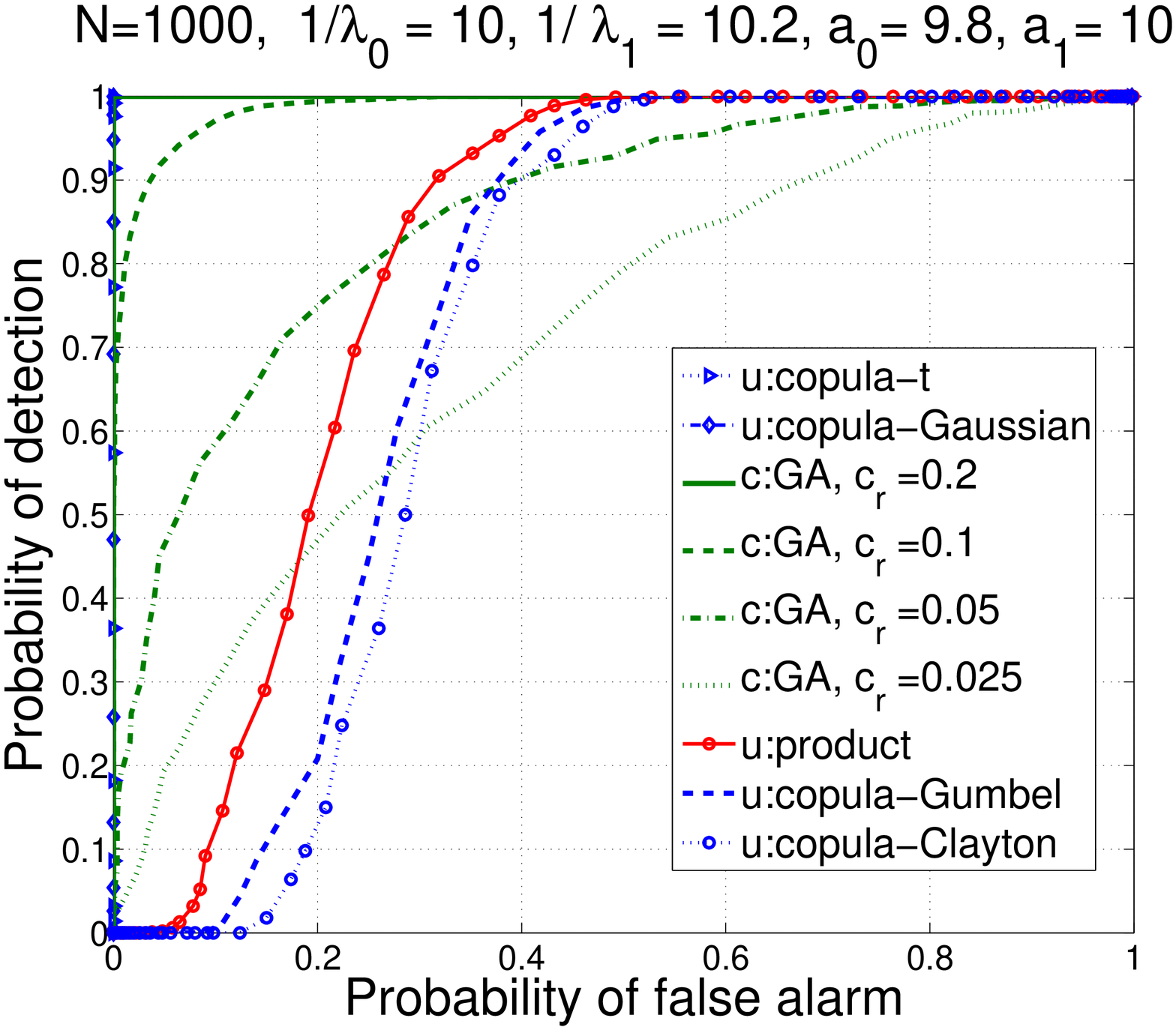}
        \caption{ Case II}
        \label{fig:N1000}
    \end{subfigure}
~
        \begin{subfigure}[b]{0.45\textwidth}
        \includegraphics[width=\textwidth]{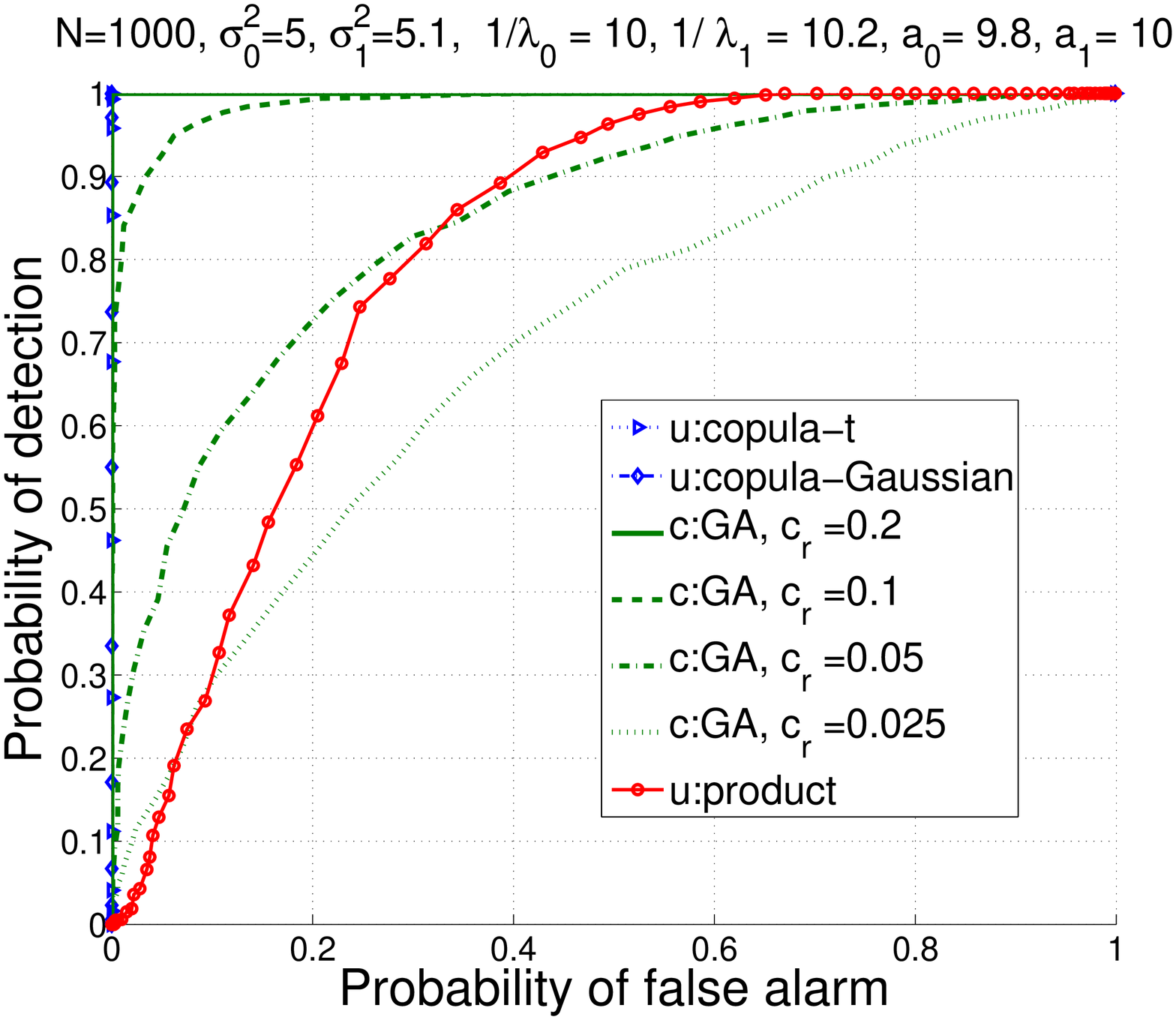}
        \caption{ Case III}
        \label{fig:N1000}
    \end{subfigure}
        \caption{Performance of  dependent  data fusion for detection in Example 1: product/copula based  approach with uncompressed data vs  Gaussian approximation with compressed data:  $N=1000$}\label{fig:ROC_case_I_II}
\end{figure}

\subsubsection*{Detection  With uncompressed data vs. detection  with compressed data via Gaussian approximation }
We compare the detection performance of LR based detection with compressed and uncompressed data. The compressed detector with Gaussian approximation, `c:GA' is compared with  the product approach (where   dependence  is ignored), `u:product', and the copula based approach, `u:copula-name', with uncompressed data. For the copula based approach, we consider Gaussian, t, Gumbel and Clayton copula functions  as described in \cite{iyengar_tsp11,He_tsp2015} for the bivariate case (Cases I and II) and Gaussian and t copula for the tri-variate case (Case III).   Fig. \ref{fig:ROC_case_I_II} shows the  performance in terms of the ROC curves for the three  cases considered in the  example. The  parameter values are provided   in figure titles.   To obtain the ROC curves, $10^3$ Monte Carlo runs were performed throughout unless otherwise specified.  With the considered parameter  values under the  two hypotheses, `u:product' does not provide perfect detection.
We make several important observations here.

 \begin{itemize}
\item For Case I where uncompressed data at the first two sensors are dependent and uncorrelated ($\mathbf D^1$ is diagonal),  `u:product',  `u:copula-t', and `u:copula-Gumbel'   perform much   better than `c:GA' even with $c_r=1$ as can be seen in Fig. \ref{fig:ROC_case_I_II}(a). In this case, with diagonal $\mathbf D^1$, existing higher-order    dependence is not taken into account in the compressed domain.
    \item For Case II where $\mathbf D^1$ is not diagonal, as can be seen in Fig. \ref{fig:ROC_case_I_II}(b), `c:GA' shows a significant performance gain  over `u:product'  after $c_r$ exceeds a certain threshold. Fusion   with `u:copula-Gaussian' and `u:copula-t'  leads to perfect detection while the  fusion performance with `c:GA' with fairly small value of $c_r$  is also capable of providing perfect detection for  the  parameter values  considered.
        \item For Case III, similar  results are seen as in Case II when `c:GA' is compared with `u:product',  `u:copula-Gaussian',  and `u:copula-t'.

           \item  In Cases II and III, dependence is taken into account via the covariance matrix in the compressed domain as discussed in Section \ref{sec_impact_dependence}. Thus, irrespective of the dimensionality reduction, due to the capability to capture  a certain amount of  dependence in the compressed domain,`c:GA' is capable of providing a significant performance gain over  `u:product' and comparable performance compared to  `u:copula-Gaussian',  and `u:copula-t'.
               \item When going from Case II to Case III (i.e., from two sensors to three sensors), `c:GA' does not show a  significant performance improvement for a given value of  $c_r$. This is because, only $\mathbf x_2$ and $\mathbf x_3$ are spatially correlated, and $\mathbf x_1$ is uncorrelated with the rest. Thus, the covariance information accounted for  in the compressed domain is the same for both cases.
\end{itemize}

We further  illustrate the behavior of the upper bound on the probability of error for Case II.   In  Fig. \ref{fig:BD_KL_case_II},  we plot the Bhattacharya distance and the upper bound on the probability of error on the left and right subplots, respectively,  for the same parameter values as in Fig. \ref{fig:ROC_case_I_II}(b).
\begin{figure}[h!]
\centerline{\epsfig{figure=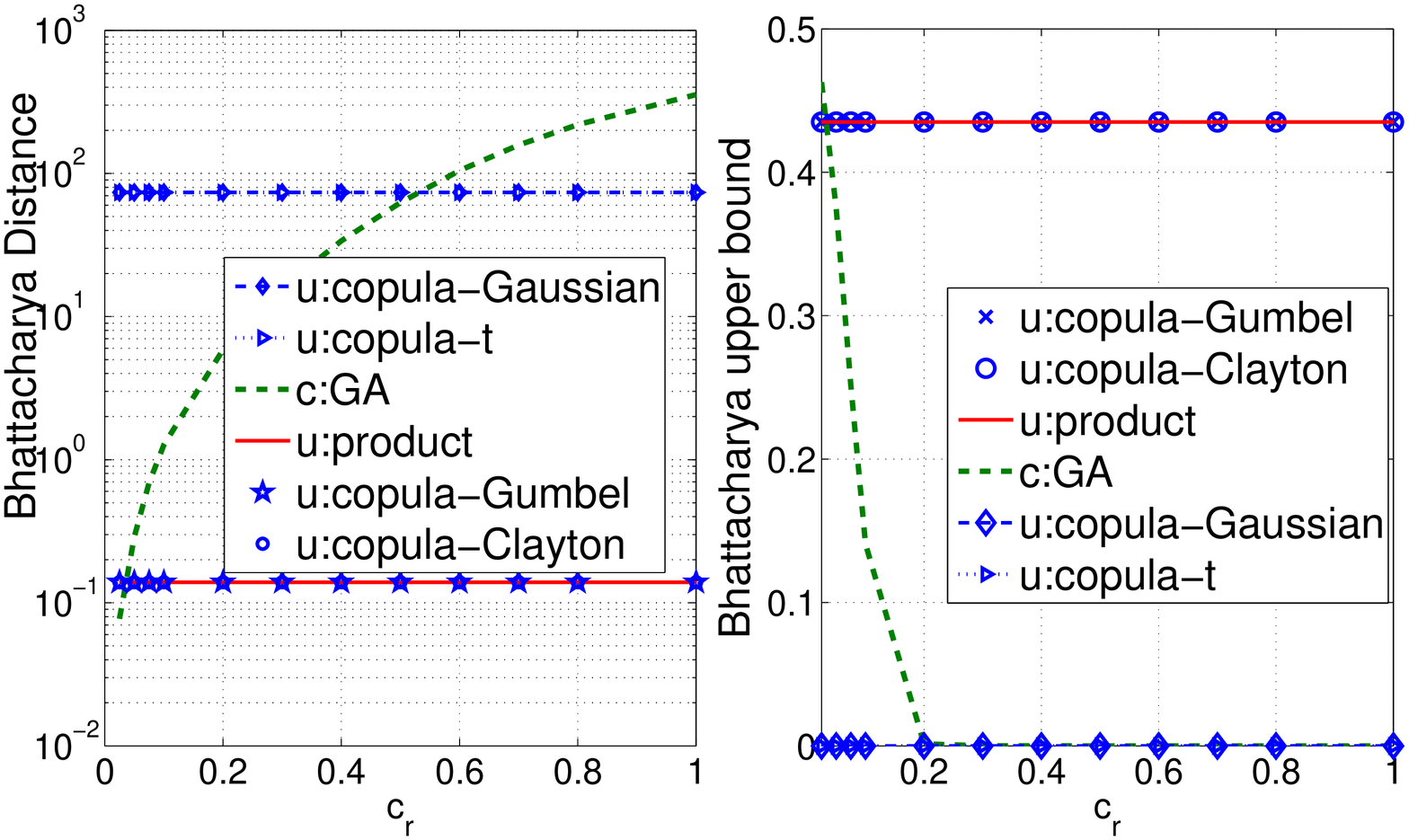,width=9.20cm}}
\caption{Bhattacharya distance and the upper on $P_e$  vs $c_r$ with compressed and uncompressed data for Case II in Example 1:  $N=1000$, $1/\lambda_0=10$, $1/\lambda_1=10.2$, $a_0=9.8$, $a_1=10$}\label{fig:BD_KL_case_II}
\end{figure}
It is seen that $\mathcal D_{B}^{\mathrm{c:GA}}$ is much larger  than $\mathcal D_{B}^{\mathrm{u:product}}$ for almost all  the values of $c_r$. Based on the distance measures shown in the left subplot, it is expected for `u:copula-Gaussian',  and `u:copula-t' to perform better than CS based detection for smaller values of $c_r$. However, as can be seen in the right subplot in Fig \ref{fig:BD_KL_case_II}, the upper bound on the probability of error with `c:GA' coincides ($\rightarrow 0$) that with `u:copula-Gaussian',  and `u:copula-t'   when $c_r$ exceeds a certain value.   This observation is intuitive since these distance measures are not linearly related to the  probability of error. Thus,  even though $\mathcal D_{B}^{\mathrm{c:GA}} < \mathcal D_{B}^{\mathrm{u:gvn}}$ with a given suboptimal approach, compressed detection via `c:GA' can be promising, as discussed in Subsection \ref{sec_Bhatt}.

\begin{figure}[h!]
    \centering
     \begin{subfigure}[b]{0.4\textwidth}
        \includegraphics[width=\textwidth]{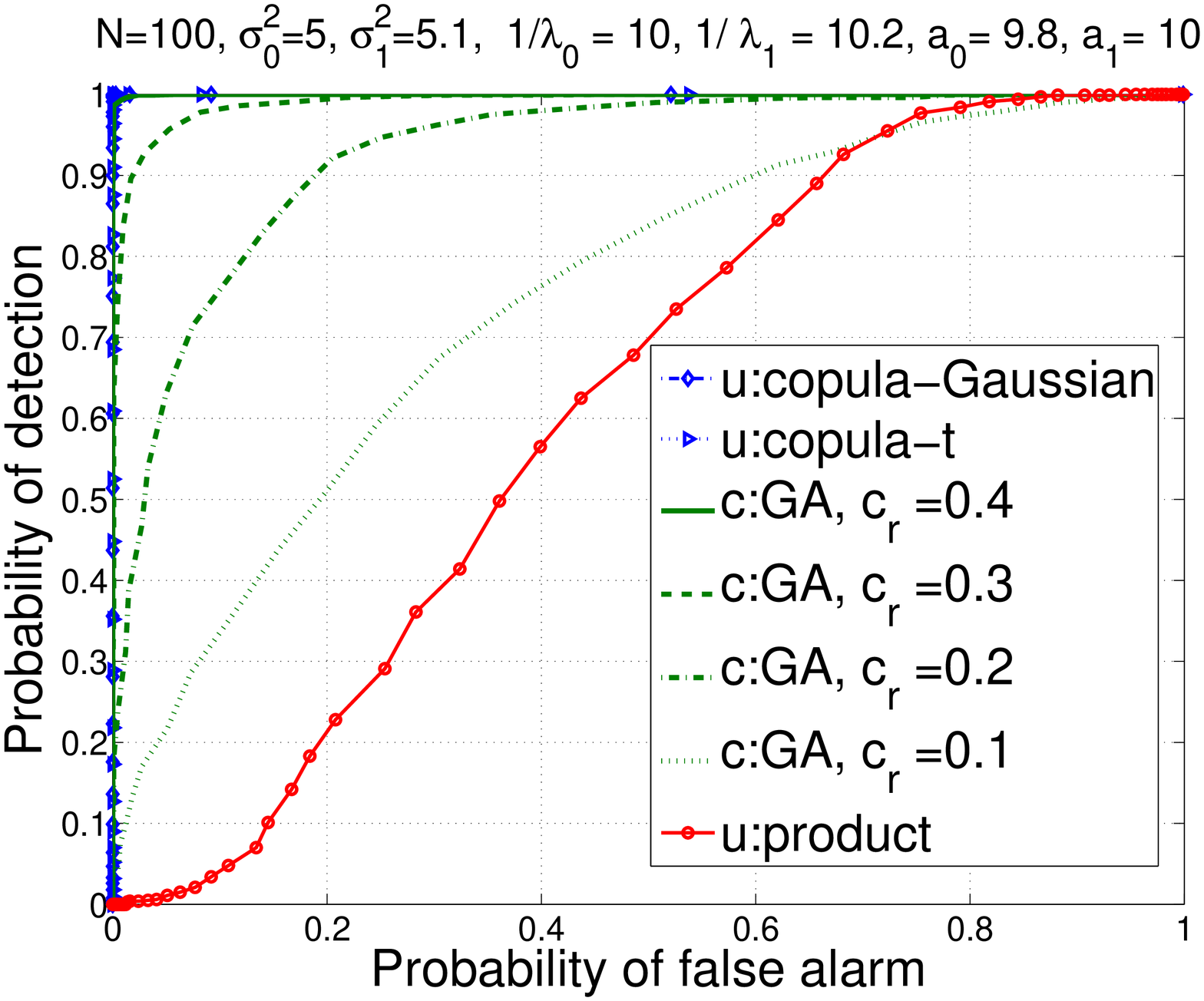}
        \caption{ Case III}
        \label{fig:N1000}
    \end{subfigure}
    ~
    \begin{subfigure}[b]{0.4\textwidth}
        \includegraphics[width=\textwidth]{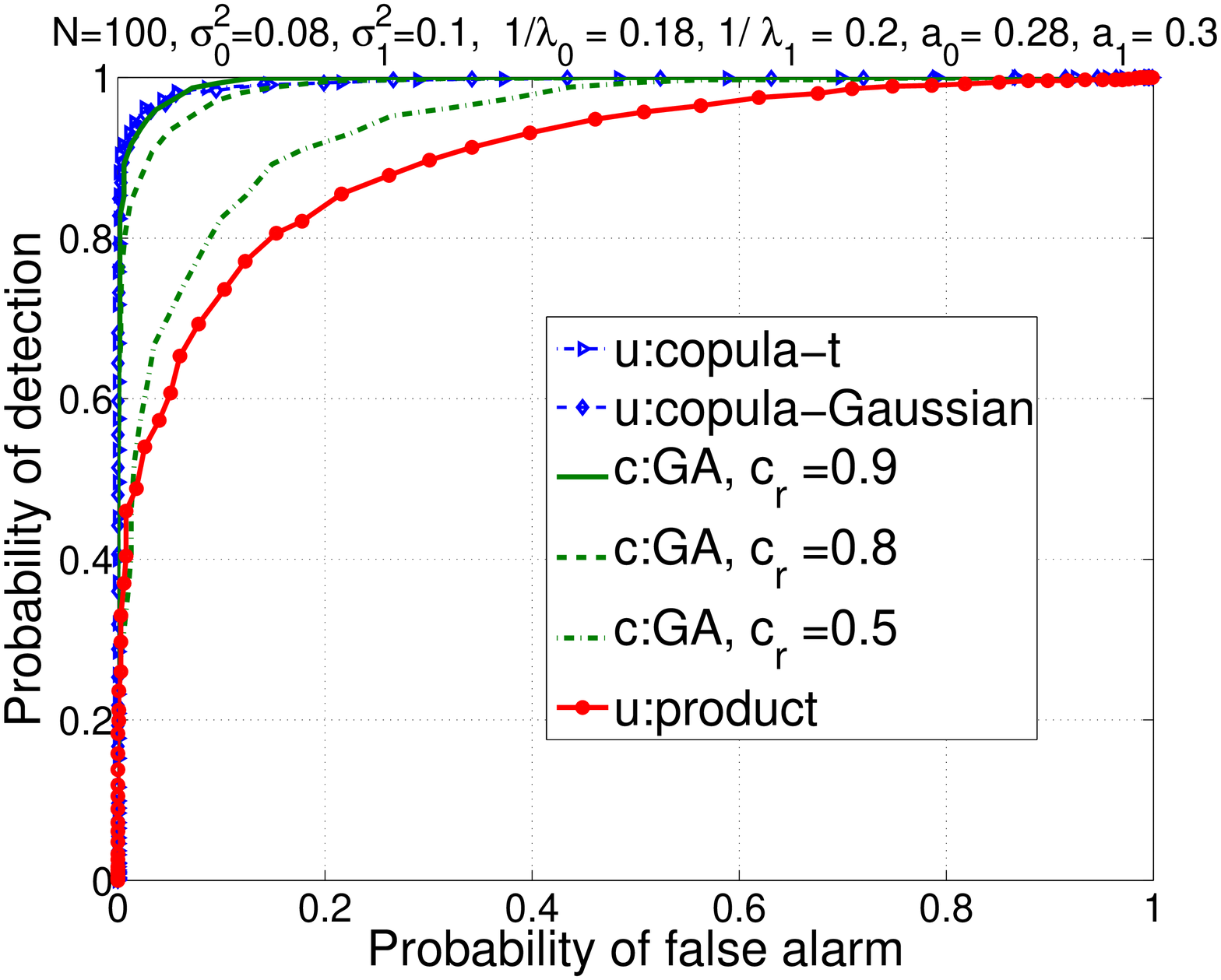}
        \caption{ Case III}
        \label{fig:N100}
    \end{subfigure}
       \caption{Performance of  dependent  data fusion for detection in Example 1: product/copula based  approach with uncompressed data vs  Gaussian approximation with compressed data; $N=100$ }\label{fig_copula_v2}
\end{figure}
Since the difference in performance  among different approaches varies  as the parameter values of the statistics under two hypotheses change, in Fig. \ref{fig_copula_v2}, we show the ROC curves with another set of parameter values considering  Case III for $N=100$.  It can be seen that  when the statistics of the data under the two hypotheses are such that `u:product' does not provide perfect detection, `c:GA' outperforms `u:product' when $c_r$ exceeds a certain threshold. Further, as $c_r$ increases,  `c:GA' shows  similar performance  as with `u:copula-t' and `u:copula-Gaussian'. Results in Fig. \ref{fig_copula_v2} again verify  that the amount of dependence captured in the compressed domain via Gaussian approximation leads to better detection  performance than `u:product' and similar performance as with the copula based approaches.

\subsubsection{Example 2}
In the second example,  we consider  the  detection of a signal in the presence of noise with $L=2$ where the signals  of interest at the two nodes are (spatially)  dependent of each other under $\mathcal H_1$. The  model for heterogeneous  uncompressed sensor data is given by
\begin{eqnarray}
\mathcal H_1&:& \mathbf  x_j = \mathbf s_j + \mathbf v_j\nonumber\\
\mathcal H_0&:& \mathbf  x_j = \mathbf v_j
\end{eqnarray}
for $j=1,2$. The noise vector $\mathbf v_j$ is assumed to be Gaussian with mean vector $\mathbf 0$ and covariance matrix $\sigma_v^2 \mathbf I$. We assume that the $n$-th elements of  $\mathbf s_1$ and $\mathbf s_2$, respectively,  are governed by a common random phenomenon so that they are dependent. For illustration purposes, we assume that the dependence model is given by:  $s_{n1} = s_n^2 + w_{n1}^2$, $s_{n2} = s_n^2 + u_{n1}^2+u_{n2}^2$, where the random variables $s_n, w_{n1}, u_{n1}, u_{n2}$ are iid Gaussian with mean zero and variance $\sigma_s^2$.  With this model, $s_{n1}\sim \mathrm{exp}(\lambda_1)$ with $\lambda_1 = \frac{1}{2 \sigma_s^2}$ and $\frac{s_{n2}}{\sigma_s^2} \sim \mathcal X_3^2$ where $x\sim \mathcal X_{\nu}^2$ denotes that $x$ has a chi-squared pdf with degree of freedom $\nu$. Then, it can be shown that the marginal pdfs of $x_{n1}$ and $x_{n2}$ are given by $f_1( x_{n1}|\mathcal H_1) = \lambda_1 e^{-\lambda_1 x_{n1}} e^{\frac{\sigma_v^2\lambda_1^2}{2}}\left(1-Q\left(\frac{x_{n1}-\sigma_v^2 \lambda_1}{\sigma_v}\right)\right)$ where $Q(\cdot)$ denotes the Gaussian $Q$ function  and $f_1( x_{n2}|\mathcal H_1) = \frac{\sqrt{\sigma_v}}{2\pi \sigma_s^3}e^{\frac{1}{8\sigma_v^2\sigma_s^4}} e^{-\frac{1}{4\sigma_v^2}\left(x_{n2}+\frac{\sigma_v^2}{2\sigma_s^2}\right)^2} G_{-3/2} \left(\frac{\sigma_v^2 - 2\sigma_s^2 x_{n2}}{2\sigma_s^2 \sigma_v}\right)$ where $G_p(z) = \frac{e^{-\frac{z^2}{4}}}{\Gamma(-p)}\int_0^{\infty} e^{-xz -\frac{x^2}{2}} x^{-p-1} dx$ with $p < 0$.  Under $\mathcal H_0$, $x_{n1}$ and $x_{n2}$ have Gaussian pdfs with mean zero and variance $\sigma_v^2$.  In this example, we have non-diagonal $\mathbf D^1$ with $\mathbf D^1_1=(\sigma_v^2+\frac{1}{\lambda_1^2})\mathbf I$, $\mathbf D_{12}^1=\mathbf D_{12}^1=2 \sigma_s^4\mathbf I$ and $\mathbf D_2^1 =  (\sigma_v^2+6\sigma_s^4)\mathbf I$.

\begin{figure}[h!]
    \centering
     \begin{subfigure}[b]{0.4\textwidth}
        \includegraphics[width=\textwidth]{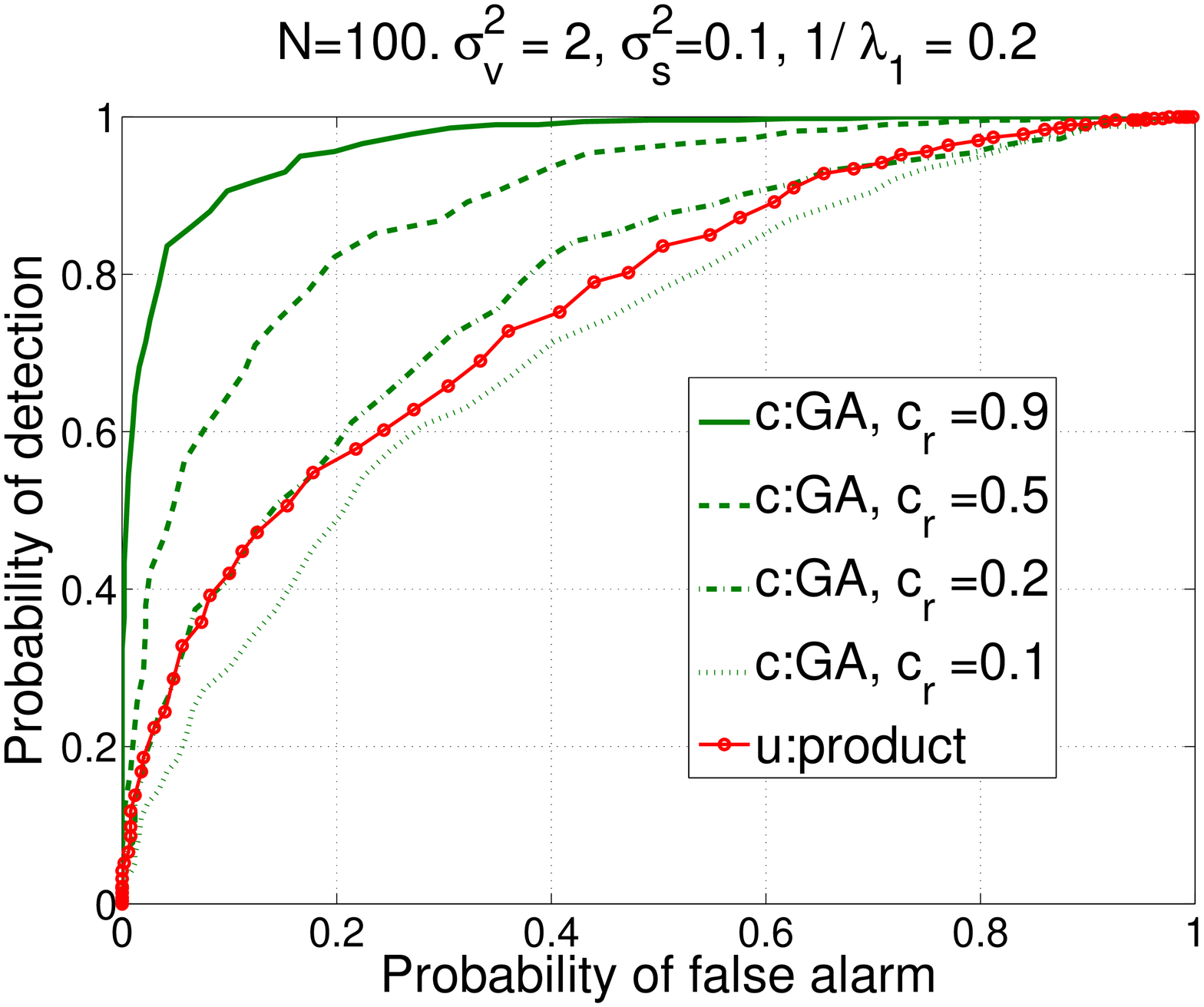}
        \caption{$N=100$,  $\sigma_v^2=2$, $\sigma_s^2=0.1$, $1/\lambda_1=0.2$}
        \label{fig:N1000}
    \end{subfigure}
    ~
    \begin{subfigure}[b]{0.4\textwidth}
        \includegraphics[width=\textwidth]{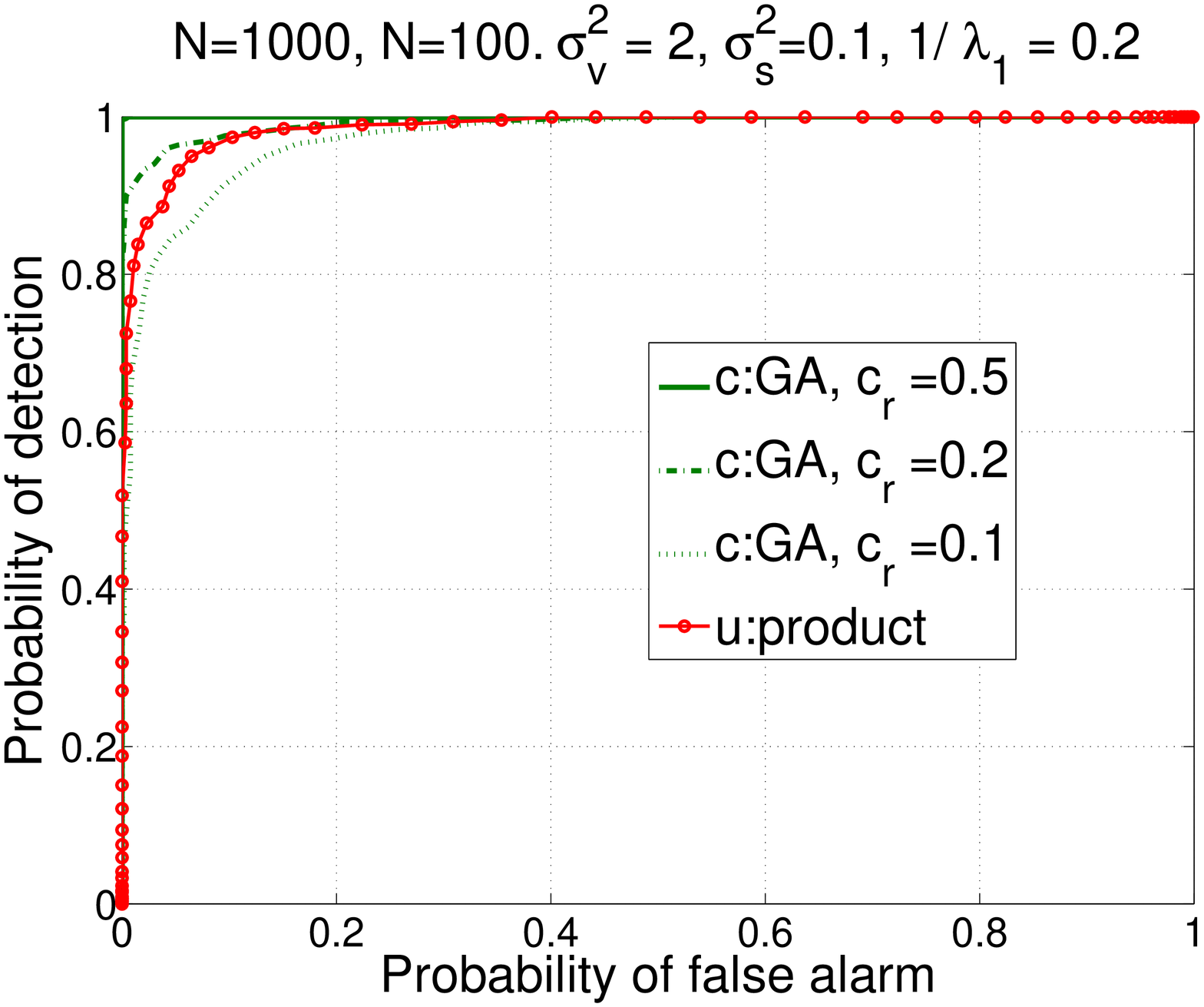}
        \caption{ $N=1000$,  $\sigma_v^2=2$, $\sigma_s^2=0.1$, $1/\lambda_1=0.2$}
        \label{fig:N100}
    \end{subfigure}
       \caption{Performance of  dependent  data fusion for detection in Example 2: Product  approach with uncompressed data vs  Gaussian approximation with compressed data}\label{fig_Eg_2}
\end{figure}
In Figs. \ref{fig_Eg_2}(a) and  \ref{fig_Eg_2}(b),  we plot the ROC curves for  $N=100$ and $N=1000$, respectively. The considered values for $\sigma_v^2$, $\sigma_s^2$ (and $\frac{1}{\lambda_1}= 2\sigma_s^2$)  are stated in the figure captions. We compare the performance of `c:GA'  with that of `u:product'. With the considered parameter values, the performance of `u:product'  with $N=100$ is not very good. However, `c:GA' performs significantly better than `u:product' as $c_r$ increases. With $N=1000$, `u:product' shows almost close to prefect detection while  similar or better performance  is achieved with `c:GA'  with a very small value for  $c_r$. As shown in Table \ref{table_example2},  `u:product'  in this example consumes a significant amount of computational power compared to `c:GA'. Fig. \ref{fig_Eg_2} and Table \ref{table_example2} verify the applicability of the proposed approach in  the presence of spatially   dependent and correlated data in terms of both performance and computational complexity.

\begin{remark}
In Examples 1 and 2, the parameter values are selected such that the mean  parameters of uncompressed data under the  two hypotheses at a given node  are not significantly different from each other. Otherwise, `u:product' can work well  since then the second or higher order statistics are not significant to distinguish between the two hypotheses. In such scenarios, efforts to  model dependence do not carry additional benefits to the fusion problem, thus such scenarios are not  of   interest in this paper.     \end{remark}

\begin{remark}
Covariance matrices, which measure the degree of linear dependence,
partially describe the dependence structure of multivariate data (when the variables are multivariate Gaussian, this description is complete).    In particular, when  the uncompressed data is non-Gaussian, dependent, and uncorrelated,   the  covariance   information is not capable of characterizing the true dependence. Thus, when such data is compressed via random projections, the dependence information is unaccounted for in the compressed domain while performing `c:GA'.
\end{remark}

\begin{remark}
When  the uncompressed  data is non-Gaussian, dependent and correlated,   the  covariance   information partially characterizes  the true dependence. In this case, when the  data is compressed via random projections, the dependence information characterized by the covariance matrix  (with a certain distortion/change)  is partially accounted for in the compressed domain while performing `c:GA'.
\end{remark}
\subsection{Computational and Communication Complexity}
With `c:GA', the computational complexity of  computing  the decision statistic \eqref{stat_Gaussian} is dominated by the computation  of ${\mathbf C^1}^{-1}$ (computation of ${\mathbf C^0}^{-1}$ is straight forward since $\mathbf C^0$ is diagonal due to spatial and temporal independence assumption under $\mathcal H_0$). Computation of ${\mathbf C^1}^{-1}$ is also straight-forward when the elements of $\mathbf x$ are uncorrelated (as considered in Case I in Example 1) since then $\mathbf C^1$ becomes diagonal. With spatially correlated uncompressed data, computation of the inverse of a $ML\times ML$ matrix is required. For $L=2$, $\mathbf C^1$ can be partitioned into $4$ blocks of each of size $M\times M$, and  the matrix inversion Lemma in block form can be exploited. This way, it is necessary to compute the inverse of a $M\times M$ matrix. For $L>2$, the block inversion Lemma can be still used with nested partitions to compute  ${\mathbf C^1}^{-1}$.  With `u:product', the likelihood ratio is computed using the given marginal pdfs. For the copula based approaches, computation of the parameters corresponding to a given copula function is required in addition to the computation of the joint marginal pdfs. The parameters that need to be computed for different copula functions considered  above are summarized  in Table I in \cite{He_tsp2015}.

For illustration, we provide in the following,  the average run time (in seconds)  required to compute the decision statistic for  Examples 1  and 2 considered above  with different approaches. For Example 1,  Cases II and III are considered with $N=100,1000$ in Table. \ref{table_example}. For Example 2, run times with $N=100$  and $N=1000$ are shown in Table \ref{table_example2}.  The run time is computed with MATLAB in a  Intel(R) Core(TM) i7-3770 CPU$@$ 3.40GHzz processor with 12 GB RAM. To estimate the parameters for each copula function, we use the 'copulafit' function and copula density was computed using the  function  'copulapdf' in Matlab.
\begin{table}[!h]
\renewcommand{\arraystretch}{1.3}
\caption{Average run time (in seconds) required to compute decision statistics in Example 1}
\label{table_example}
\centering
\begin{tabular}{|l|l|l|}
\hline
Approach & $N=100$ & $N=1000$\\
$~$ & Case II ~ Case III & Case II ~ Case III\\
$~$ & ($L=2$) ~ ($L=3$) & ($L=2$) ~ ($L=3$)\\
\hline
`u:product'  &  0.0080 ~~~~~~~~~  0.0281  &  0.0107~~~~~~~ 0.0322\\
\hline
`u:copula-Gaussian'   & 0.0105 ~~~~~~~~~  0.0314  & 0.0138~~~~~~~  0.0359\\
\hline
`u:copula-t'   &  0.0664 ~~~~~~~~~  0.0948 &0.2730 ~~~~~~ 0.3634\\
\hline
`c:GA',   $c_r=0.1$   &1.2239e-04 ~ 1.2334e-04 & 7.3375e-04 ~ 0.0016\\
\hline
`c:GA', $c_r=0.2$   &    1.4958e-04 ~ 1.5876e-04  & 0.0016 ~~~~~~ 0.0029\\
\hline
`c:GA', $c_r=0.5$   &2.4795e-04 ~ 2.5894e-04&  0.0091  ~~~~~~ 0.0097\\
\hline
`c:GA', $c_r=0.9$   & 3.0501e-04 ~ 3.5743e-04  & 0.0293 ~~~~~~  0.0294\\
\hline
\end{tabular}
\end{table}

\begin{table}[!h]
\renewcommand{\arraystretch}{1.3}
\caption{Average run time (in seconds) required to compute decision statistics in Example 2}
\label{table_example2}
\centering
\begin{tabular}{|l|l|l|}
\hline
Approach & $N=100$ & $N=1000$ \\

\hline
`u:product'  &  0.1425  &   1.4520 \\
\hline
`c:GA', $c_r=0.1$   &  4.5356e-04 & 0.0092 \\
\hline
`c:GA', $c_r=0.2$   &  6.8284e-04  & 0.0436 \\
\hline
`c:GA',  $c_r=0.5$   &   0.0027 &   0.4786\\
\hline
\end{tabular}
\end{table}

It can be seen that, `c:GA' with even fairly   large $c_r$ ($< \approx 0.5$) consumes less time than  all the approaches considered with uncompressed data for  both examples. In Example 2 where the  marginal pdfs are not readily available to compute the decision statistic for `u:product', this gap in run times becomes more significant.

In terms of the communication  overhead, to perform all the suboptimal  approaches considered with uncompressed data, each node is required to transmit its length-$N$ observation vector to the fusion center. On the other hand, with `c:GA', each node is required to transmit only a length-$M$ ( $M < N$) vector. Thus, the communication overhead, in terms of the total number of messages to be transmitted by each node,  is reduced by a factor of $c_r=\frac{M}{N}$ with `c:GA' compared to all the approaches with uncompressed data.

In summary, from Examples 1 and 2, we can conclude the  following:
\begin{itemize}

         \item The computational complexity  of `c:GA' with small $c_r$ is significantly  less than that with the other approaches with uncompressed data.
               \item The communication overhead of `c:GA' is reduced by a factor $c_r$ compared to all the other approaches with uncompressed data.
                   \item `c:GA' performs significantly better than `u:product' and shows  similar/comparable  performance as  that with the  copula based approach when the covariance matrix of uncompressed data under $\mathcal H_1$  is   non-diagonal  given that the mean parameters of data under the two hypotheses are not significantly different from each other.
       \item With dependent but uncorrelated uncompressed data, `c:GA' can  still be advantageous in terms of the computational/communication  complexity at the expense  of  a small loss of performance compared to `u:product'  and  quite significant performance loss compared to the copula based approaches.

        \end{itemize}

\section{Detection   with  Compressed Dependent  Data Based on Second Order Statistics of Uncompressed Data}\label{sec_covariance}
In Section \ref{sec_likelihood}, the detection problem was solved in the compressed domain assuming that the marginal pdfs and the first and  second order statistics of the uncompressed data under each  hypothesis  are known  (or can be accurately estimated).  However,  these  assumptions may be too restrictive in practical settings. In the following, we consider a nonparametric approach where the goal is to compute a  decision statistic for detection  based on the  statistics of uncompressed data where such statistics are estimated from compressed measurements.  Consider that  each node (modality) has access to a stream of data $\mathbf x_j(t)$ for $t=1,\cdots,T$. Further let $\mathbf x(t)=[\mathbf x_1(t)^T, \cdots, \mathbf x_L(t)^T]^T$ and   redefine $\mathbf D_x$ to be   the covariance matrix of $\mathbf x(t)$.  We consider a  decision statistic of the form
\begin{eqnarray*}
\Lambda_{\mathrm{cov}} = f(\mathbf D_x).
\end{eqnarray*}
Under  $\mathcal H_0$, $\mathbf D_x$ is diagonal with the assumption that the  data is independent across time and space. Under $\mathcal H_1$, $\mathbf D_x$ can have off-diagonal elements in the presence of  spatially correlated multimodal data. Since the covariance matrix has different structures under the two hypotheses,  a  decision statistic based on uncompressed covariance matrix can be computed. There are  several covariance based decision statistics computed  in \cite{Zeng_C2007,Zeng_VT09}. Covariance absolute value (CAV) detection
 is  considered in \cite{Zeng_C2007,Zeng_VT09}. With CAV, the decision  statistic becomes
\begin{eqnarray}
\Lambda_{\mathrm{cov}} = \frac{\underset{i}{\sum}\underset{j}{\sum}|  \mathbf D_x[i,j]|}{\underset{i}{\sum}|\mathbf D_x[i,i]|}. \label{Lamda_1}
\end{eqnarray}
With this statistic, when there are off-diagonal elements in the covariance matrix, we have $\Lambda_{\mathrm{cov}} > 1$ while $\Lambda_{\mathrm{cov}}=1$ when the off diagonal elements are zeros.  The goal  is to get an approximation to $\Lambda_{\mathrm{cov}}$ based on compressed data $\mathbf y(t)=\mathbf A \mathbf x(t)$ for $t=1,\cdots,T$ where $\mathbf A$ is as defined in (\ref{eq_A}).

The  covariance matrix of $\mathbf y(t)$, $\mathbf C_y$, can be expressed as
\begin{eqnarray*}
\mathbf C_y = \mathbf A \mathbf D_x \mathbf A^T.
\end{eqnarray*}
Note that, $\underset{i}{\sum} \mathbf D_x[i,i] = \mathrm{tr}(\mathbf D_x)$. We have
 \begin{eqnarray*}
 \mathrm{tr}(\mathbf C_y) &=& \mathrm{tr}(\mathbf A \mathbf D_x \mathbf A^T)= \mathrm{tr}(\mathbf A^T \mathbf A \mathbf D_x ).
 \end{eqnarray*}
 When $\mathbf A_j$ is selected as an orthoprojector for $j=1,\cdots,L$, we may approximate $\mathbf A^T \mathbf A \approx \frac{M}{N} \mathbf I$. Then, we have
  \begin{eqnarray*}
 \mathrm{tr}(\mathbf C_y) \approx  \frac{M}{N}\mathrm{tr}(\mathbf D_x ),
 \end{eqnarray*}
and, thus, $\mathrm{tr}(\mathbf D_x ) = \frac{N}{M} \mathrm{tr}(\mathbf C_y)$.  Here we approximate $\mathbf C_y $ by the  sample covariance  matrix computed as
\begin{eqnarray}
\tilde{\mathbf C}_y = \frac{1}{T}\sum_{t=1}^T [\mathbf y(t)-\mathbb E(\mathbf y(t))][\mathbf y(t)-\mathbb E(\mathbf y(t))]^T.\label{cov_sample}
\end{eqnarray}  Then, the decision statistic (\ref{Lamda_1}) reduces to
\begin{eqnarray}
\Lambda_{\mathrm{cov}} = \frac{\eta + 2\sum_{i=1}^{NL-1} \sum_{j=i+1}^{NL}  |\mathbf D_x[i,j]|}{\eta}\label{Lamda_1_2}
\end{eqnarray}
where $\eta= \frac{N}{M } \mathrm{tr}(\tilde{\mathbf C}_y)$.

The goal is to estimate the off-diagonal elements of $\mathbf D_x$ based on $\tilde{\mathbf C}_y$.  It is noted that estimation of the complete covariance matrix, $\mathbf D_x$,  is not necessary to construct $\Lambda_{\mathrm{cov}}$ in \eqref{Lamda_1_2}.
 In the case where only spatial samples are dependent and the time samples of each modality are independent, the covariance matrix has only a $2(L-1)$ diagonals (in addition to the main diagonal) with nonzero elements.
 In the following, we describe a procedure to compute $\Lambda_{\mathrm{cov}}$ in (\ref{Lamda_1_2}) based on $\tilde{\mathbf C}_y$ when there is spatial dependence of data so that $\mathbf D_x$ has a known structure.
Note that we may write $\tilde{\mathbf C}_y$ as
\begin{eqnarray}
\tilde{\mathbf C}_y &=& \underset{i,j}{\sum} \mathbf D_x[i,j] \mathbf a_i \mathbf a_j^T.
\end{eqnarray}
Let $\mathcal U$ be  a set consisting of $(i,j)$ pairs corresponding to the desired off-diagonal elements in the upper (or lower) triangle of $\mathbf D_x$. The $m$-th pair in $\mathcal U$ is denoted by, $(\mathcal U(m,1), \mathcal U(m,2))$ and $\tilde N= |\mathcal U|$.
\begin{proposition}
Let $\mathbf d_{U}$ be the vector containing elements  $\mathbf D_x[i,j] $ for $(i,j)\in \mathcal U$. The  least squares solution of  $ \mathbf d_{U}$ is given by
\begin{eqnarray}
\hat {\mathbf d}_U = \mathbf B^{-1} \mathbf b\label{LS_solution_general}
\end{eqnarray}
 where $\mathbf B$ is a $\tilde N\times \tilde N$ matrix whose $(m,r)$-th element is given by
 \begin{eqnarray}
 \mathbf B[m,r] = \mathbf a_{\mathcal U(r,2)}^T \mathbf a_{\mathcal U(m,2)} \mathbf a_{\mathcal U(m,1)}^T \mathbf a_{\mathcal U(r,1)}\label{LS_B}
 \end{eqnarray}
 and $\mathbf b$ is a $\tilde N\times 1$ vector with
 \begin{eqnarray}
 \mathbf b = [\mathbf a_{\mathcal U(1,2)}^T\tilde{\mathbf C}_y^T \mathbf a_{\mathcal U(1,1)}, \cdots, \mathbf a_{\mathcal U(\tilde N,2)}^T\tilde{\mathbf C}_y^T \mathbf a_{\mathcal U(\tilde N,1)}]^T.\label{LS_b}
 \end{eqnarray}
 \end{proposition}

\begin{proof}
Let $\mathbf R = \tilde{\mathbf C}_y -  \underset{(i,j)\in \mathcal U}{\sum} \mathbf D_x[i,j] (\mathbf a_i \mathbf a_j^T + \mathbf a_j \mathbf a_i^T) = \tilde{\mathbf C}_y - \sum_{m=1}^{\tilde N} \mathbf d_{U}[m]\tilde{\mathbf A}_m$ where $\tilde{\mathbf A}_m = \mathbf a_{\mathcal U(m,1)} \mathbf a_{\mathcal U(m,2)}^T + \mathbf a_{\mathcal U(m,2)} \mathbf a_{\mathcal U(m,1)}^T$.
The LS solution of $\mathbf d_{U}$ is found by solving
 \begin{eqnarray}
\hat{\mathbf d}_{U} =\underset{\mathbf d_{U}}{\arg\min} ||\mathbf R||_F^2 = \underset{\mathbf d_{U}}{\arg\min} ~ \mathrm{tr}(\mathbf R \mathbf R^T).  \label{LS_problem}
 \end{eqnarray}
We can  express $\mathrm{tr}(\mathbf R \mathbf R^T)$ as,
\begin{eqnarray}
\mathrm{tr}(\mathbf R \mathbf R^T) = \mathrm{tr}(\tilde{\mathbf C}_y \tilde{\mathbf C}_y ^T) - 2 \tilde{\mathbf b}^T\mathbf d_U + \mathbf d_U^T \tilde{\mathbf B} \mathbf d_U\label{RRt}
\end{eqnarray}
where $\tilde{\mathbf b}[m] = \mathrm{tr}(\tilde{\mathbf C}_y  \tilde{\mathbf A}_m^T)$ for $m=1,\cdots, \tilde N$ and $\tilde{\mathbf B} [m,r] = \mathrm{tr}(\tilde{\mathbf A}_m\tilde{\mathbf A}_m^T)$ for $m,r=1,\cdots, \tilde N$. Taking the derivative of \eqref{RRt} with respect to $\mathbf d_{U}$, $\hat{\mathbf d}_{U}$ is found as
\begin{eqnarray}
\hat{\mathbf d}_{U} = \tilde{\mathbf B}^{-1} \tilde{\mathbf b}.
\end{eqnarray}
It can be easily shown that $\tilde{\mathbf B} = 2\mathbf B$ and $\tilde{\mathbf b} = 2{\mathbf b}$ where  ${\mathbf B}$ and $\mathbf b$ are as defined in \eqref{LS_B} and \eqref{LS_b}, respectively, resulting in \eqref{LS_solution_general} which completes the proof.
\end{proof}

Then,  $\Lambda_{\mathrm{cov}}$ in \eqref{Lamda_1_2} reduces to
\begin{eqnarray}
\Lambda_{\mathrm{cov}} = \frac{\eta + 2||\hat {\mathbf d}_U||_1}{\eta}.\label{Lamda_1_BiV}
\end{eqnarray}

\subsection{Illustrative Example}
 To illustrate the detection performance with the test statistic $\Lambda_{\mathrm{cov}}$ in \eqref{Lamda_1_BiV}, we consider Example  1  given in Section \ref{example} with Case II, in which the goal is to detect the change of the statistical parameters  due to a common random phenomenon.  Uncompressed data is generated based on  the considered marginal pdfs under the two hypotheses and the dependence model under $\mathcal H_1$ the as considered in  Case II in  Example  1.
Detection using \eqref{Lamda_1_BiV} is performed
assuming that the second order statistics of uncompressed data are not  known under any hypothesis  and estimating them with compressed measurements.  For this case, there are  only two nonzero off diagonals of $\mathbf D_x$ in which the values are the same (say $d_U$). To compute $\Lambda_{\mathrm{cov}}$, estimation of only  $d_U$ is required which is given by  $
\hat { d}_U = \frac{\mathbf b^T\mathbf 1}{\mathbf 1^T \mathbf B \mathbf 1}
$ where   $\mathbf B$ and $\mathbf b$ are as defined in \eqref{LS_B} and \eqref{LS_b}, respectively.

\begin{figure}[h!]
    \centering
    \begin{subfigure}[b]{0.4\textwidth}
        \includegraphics[width=\textwidth]{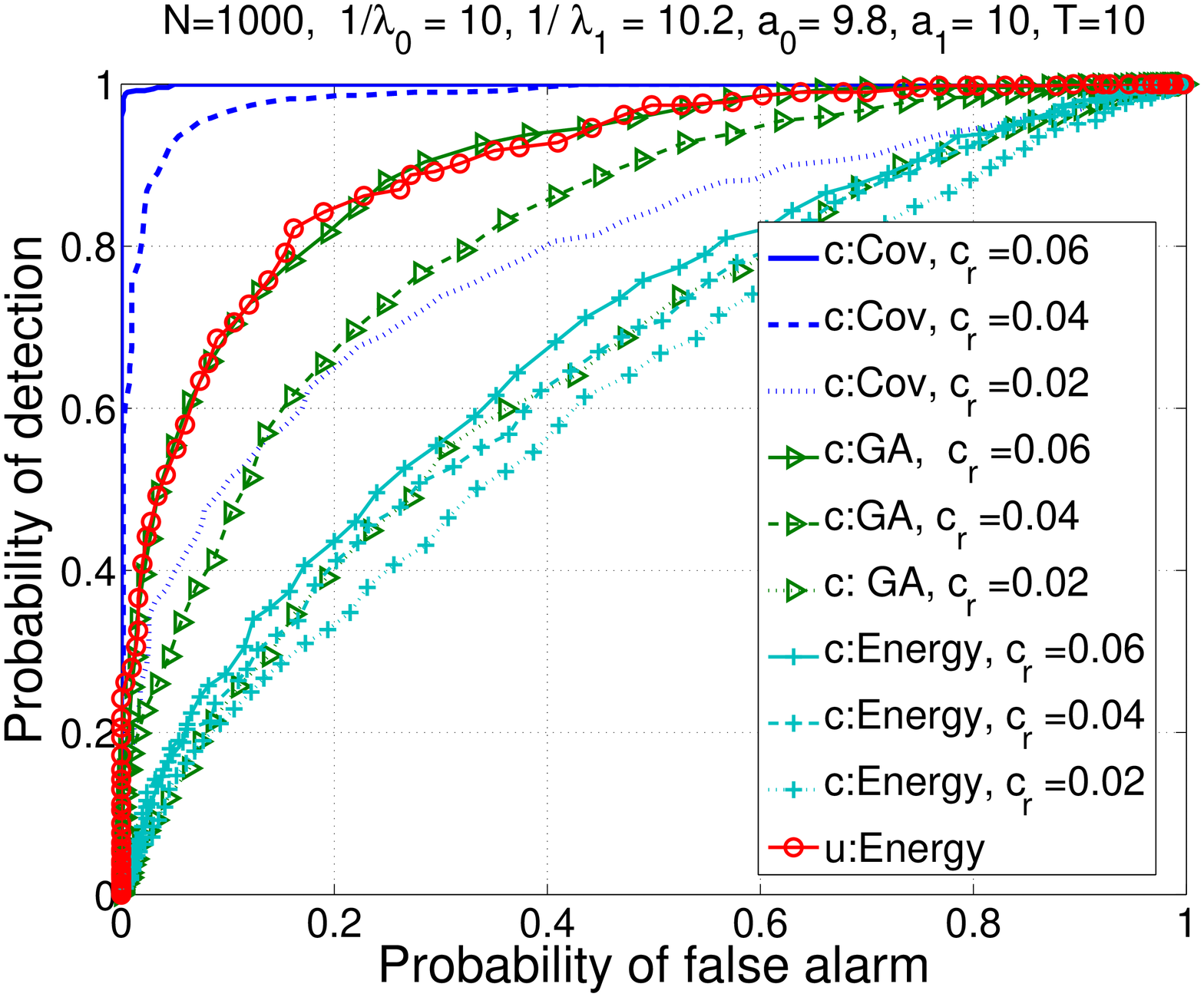}
        \caption{$c_r$ varies, $T=10$}
        \label{fig:N100}
    \end{subfigure}
    ~ 
    \begin{subfigure}[b]{0.4\textwidth}
        \includegraphics[width=\textwidth]{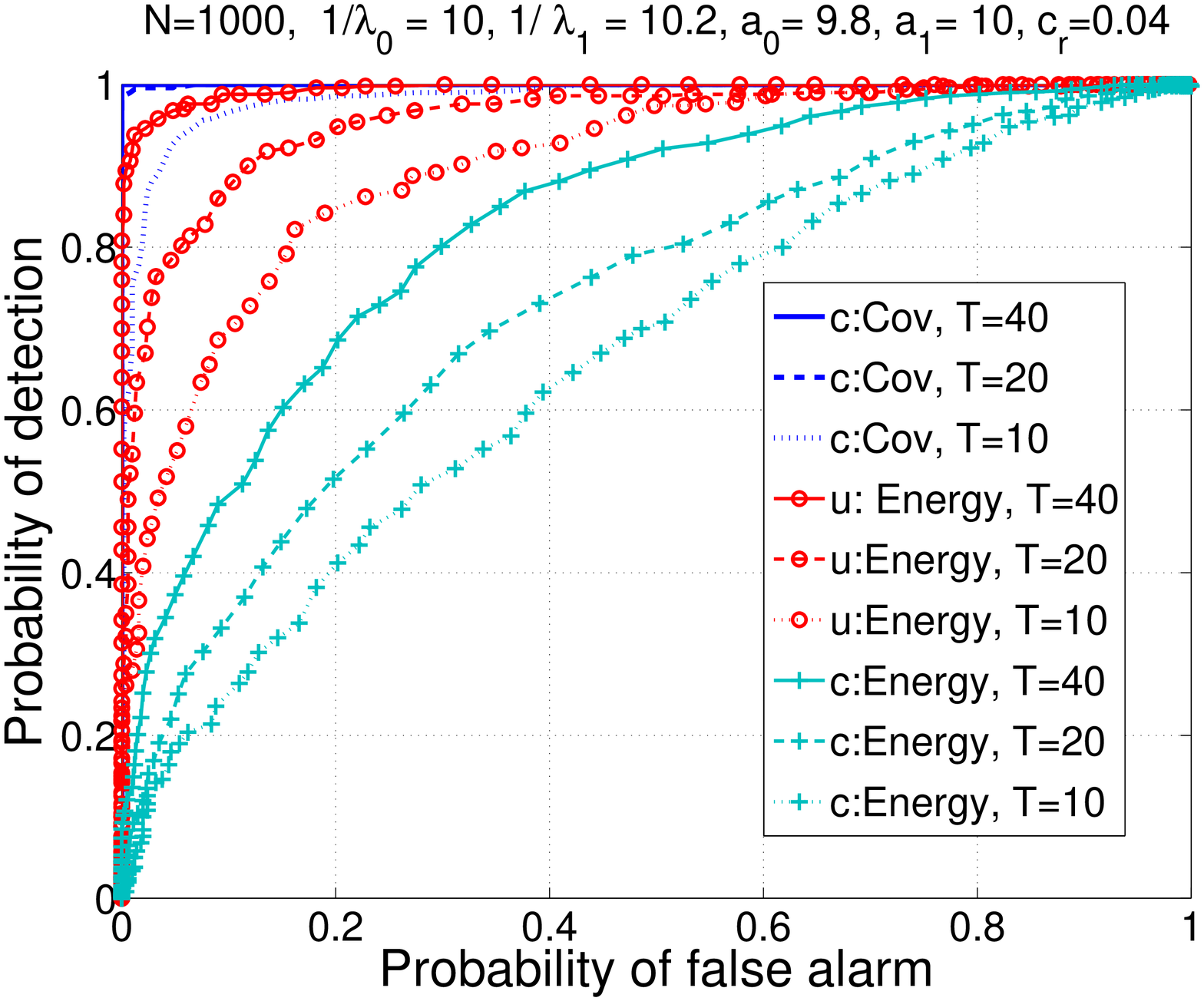}
        \caption{ $T$ varies, $c_r=0.04$}
        \label{fig:N1000}
    \end{subfigure}
       \caption{Detection performance with the test statistic \eqref{Lamda_1_BiV}; $N=1000$}\label{fig:cov_Gauss}
\end{figure}
In Fig. \ref{fig:cov_Gauss}, we plot ROC curves  with the test statistic  (\ref{Lamda_1_BiV}) for different values for $c_r$ and $T$ where $T$, as defined earlier,  is the number of sample vectors available for each modality. The detector with the test statistic  (\ref{Lamda_1_BiV}) is denoted as `c:Cov'. The parameters used to generate data under the two hypotheses are the same as  were used in Fig. \ref{fig:ROC_case_I_II} (b). ROC curves are generated using $1000$ Monte Carlo runs. In Fig. \ref{fig:cov_Gauss} (a), the  performance of `c:Cov' is shown for different values of $c_r$ keeping  $T=10$.  We  compare the results obtained  using the energy detector with compressed as well  as uncompressed data, denoted by `c:Energy', and `u:Energy', respectively, which is a widely used nonparametric detector. The test statistic of `u:Energy' and `c:Energy' is given by $\Lambda_{\mathrm{u:Energy}} = \sum_{t=1}^T ||\mathbf x(t)||_2^2$, and $\Lambda_{\mathrm{c:Energy}} = \sum_{t=1}^T ||\mathbf y(t)||_2^2$, respectively.  Further the detection performance with  Gaussian approximation, `c:GA', is also shown which assumes  that the statistics of uncompressed data are known. With the parameter values considered, for given $c_r$ and $T$, it is seen from Fig. \ref{fig:cov_Gauss} (a) that `c:Cov'  significantly outperforms `c:Energy' and  `c:Cov'. Compared to `u:Energy',  `c:Cov' outperforms `u:Energy'  after  $c_r$ exceeds a certain value (which is very small).  In Fig. \ref{fig:cov_Gauss} (b), the detection performance is shown as $T$ varies for $c_r=0.04$ so  that $M=40$. As can be seen in Fig. \ref{fig:cov_Gauss} (b),  detection performance  improves as $T$ increases for all the detectors. For `c:Cov', the estimate of $\hat { d}_U$ becomes  more  accurate as $T$ increases. However, the value of $T$ that is capable of providing almost perfect detection with `c:Cov'  is not very large compared to $M$.  As can be seen in Fig. \ref{fig:cov_Gauss} (b), almost perfect detection is achieved when $T=20$ for the parameter values considered which is less than $M$. Further,  `c:Cov' outperforms `u:Energy'  and  `c:Energy' for all the values of $T$ considered while  the performance gain achieved by `c:Cov' is significant compared to `c:Energy'.

Next, we illustrate the robustness of `c:Cov'. A CAV based test statistic as in \eqref{Lamda_1} has been used to detect a signal in the presence of Gaussian noise in \cite{Zeng_C2007,Zeng_VT09}   without any compression. It has been shown that the threshold required to keep the  probability of false alarm, $P_f$, under  a desired value, $\alpha_0$, is independent of the noise power making the CAV based detector more robust than the energy detector against the uncertainties of the noise power. With the CAV based test statistic computed in this paper  based on the compressed data as in \eqref{Lamda_1_BiV}, the computation of the threshold, $\tau_C$,  in closed-form  so that $P_f\leq \alpha_0$ is computationally intractable. In the above example, uncompressed  data under $\mathcal H_0$ is non-Gaussian and the marginal pdfs of data at the two sensors are parameterized by $\lambda_0$ and $a_0$, respectively. In Fig. \ref{fig:threshold},   we plot the threshold vs $\lambda_0$ and $a_0$ to ensure  $P_f \leq \alpha_0$ keeping $N$ are $T$ are fixed. With `u:Energy', $\Lambda_{\mathrm{u:Energy}} | \mathcal H_0$ can be approximated by a Gaussian random variable with mean $\mu_{u,E}$ and variance $\sigma_{u,E}^2$ as $NT$ is sufficiently large where $\mu_{u,E} = NT \left(\frac{2}{\lambda_0^2} + \frac{a_0}{a_0+2}\right)$ and $\sigma_{u,E}^2 = NT \left(\frac{20}{\lambda_0^2} + \frac{4a_0}{(a_0+4)(a_0+2)^2}\right)$. Then, the threshold, $\tau_{u,E}$,  so that $P_f\leq \alpha_0$ can be obtained as $\tau_{u,E} = NT \left(\frac{2}{\lambda_0^2} + \frac{a_0}{a_0+2}\right) + Q^{-1}(\alpha_0) \sqrt{NT\left(\frac{20}{\lambda_0^2} + \frac{4a_0}{(a_0+4)(a_0+2)^2}\right)}$ where $Q^{-1}(\cdot)$ denotes the inverse Gaussian $Q$ function. Similarly with `c:Energy', the threshold $\tau_{c,E}$ can be found as $\tau_{c,E} = MT \left(\frac{2}{\lambda_0^2} + \frac{a_0}{a_0+2}\right) + Q^{-1}(\alpha_0) \sqrt{MT\left(\frac{20}{\lambda_0^2} + \frac{4a_0}{(a_0+4)(a_0+2)^2}\right)}$.

In Fig. \ref{fig:threshold}, the threshold required to keep $P_f\leq \alpha_0$ with `c:Cov'` and `u:Energy' vs $a_0$ and $\lambda_0$ is  shown for given $N$ and $T$. It can be observed that the variation of $\tau_C$  with respect to $a_0$ and $\lambda_0$ is fairly  small (negligible). However,   $\tau_{u:E}$ varies significantly as  $a_0$ and $\lambda_0$ vary (similar observations are seen for $\tau_{c:E}$ while the figures are not included for brevity).  Thus, in addition to the performance gain achieved  over `u:Energy' (and `c:Energy'), `c:Cov' is more robust against the uncertainties of the signal parameters under $\mathcal H_0$ than the energy detector.

\begin{figure}[h!]
    \centering
    \begin{subfigure}[b]{0.4\textwidth}
        \includegraphics[width=\textwidth]{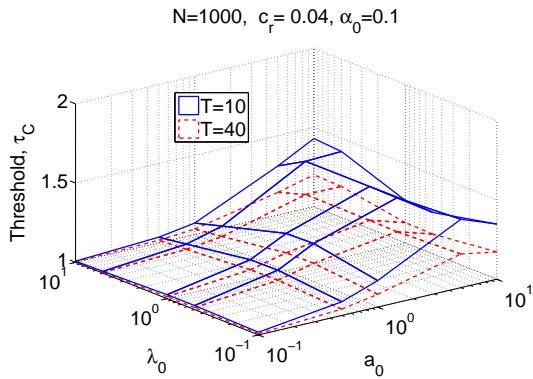}
        \caption{`c:Cov'}
        \label{fig:N100}
    \end{subfigure}
    \begin{subfigure}[b]{0.4\textwidth}
        \includegraphics[width=\textwidth]{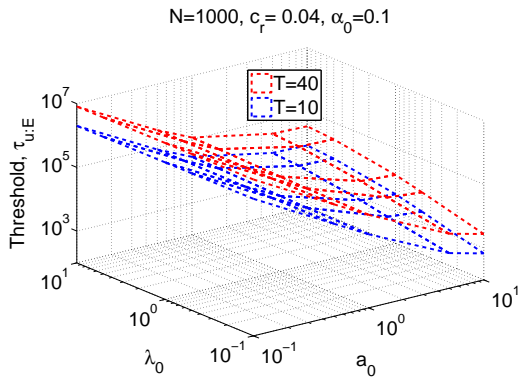}
        \caption{`u:Energy'}
        \label{fig:N1000}
    \end{subfigure}
       \caption{Threshold of `c:Cov' and `u:Energy' vs $a_0$ and $\lambda_0$}\label{fig:threshold}
\end{figure}

\subsection{Computational Complexity}
For $\Lambda_{\mathrm{cov}}$, the computation of $\mathbf B^{-1}$  and $\mathbf b$ as in \eqref{LS_B} and \eqref{LS_b}, respectively is required  in addition to computing the sample covariance matrix $\tilde {\mathbf C}_y$. Computation of  $\Lambda_{\mathrm{u:Energy}}$ and $\Lambda_{\mathrm{c:Energy}}$ is straight  forward from $\mathbf x(t)$ and $\mathbf y(t)=\mathbf A \mathbf x(t)$, respectively, for $t=1,\cdots,L$.  The run times  required to compute the decision statistics for the four  approaches considered in Fig. \ref{fig:cov_Gauss}(a) are  listed in Table \ref{table_example3} when the input is given as $\mathbf x(t)$ for $t=1,\cdots,L$. The statistic of `c:GA' is independent of $T$ since we assume perfect knowledge of the covariance matrix of uncompressed data for the Gaussian approximation based approach.  It is noted that, the decision statistic was computed over $10^3$ trials  to get the average run time. From Table \ref{table_example3}, it can be observed that a relatively  large run time is required for `c:Cov' compared to the other  approaches. This is the price to pay for the performance  gain achieved as depicted in Fig. \ref{fig:cov_Gauss} (a) and the robustness in threshold setting against the signal parameters under $\mathcal H_0$ as depicted in Fig. \ref{fig:threshold}. Further, in `c:Cov',  the run time does not significantly increase when  $T$ increases (going from  $T=10$ to $T=40$) although this increase in $T$ can improve the performance as can be seen in Fig. \ref{fig:cov_Gauss} (b).
\begin{table}[!h]
\renewcommand{\arraystretch}{1.3}
\caption{Average run time (in seconds) required to compute decision statistics for `c:Cov', `c:GA', `c:Energy' and `u:Energy' for Case II in  Example 2}
\label{table_example3}
\centering
\begin{tabular}{|l|l|l|}
\hline
Approach & $N=1000$ &  $N=1000$ \\
$~$& $T=10$ & $T=40$\\
\hline
`u:Energy'    &   6.1529e-04 &  0.0027 \\
\hline
`c:Energy' $c_r=0.02$   &  3.4280e-04  & 7.7378e-04 \\
\hline
`c:Energy' $c_r=0.04$   &    4.5460e-04  &  0.0011 \\
\hline
`c:Energy' $c_r=0.06$   &    5.3973e-04 & 0.0014\\
\hline
\hline
`c:GA'  $c_r=0.02$    &   4.0746e-04 &   4.0746e-04 \\
\hline
`c:GA'  $c_r=0.04$    &   4.9580e-04 &   4.9580e-04 \\
\hline
`c:GA' $c_r=0.06$     & 5.5683e-04 & 5.5683e-04\\
\hline
\hline
`c:Cov' $c_r=0.02$    &  0.0137 &    0.0148  \\
  \hline
`c:Cov' $c_r=0.04$    &  0.0163 &   0.0178 \\
\hline
`c:Cov'  $c_r=0.06$     &  0.0258  &   0.0277 \\
\hline
\end{tabular}
\end{table}

\section{Experimental  Results with Real  Data}\label{sec_simulation}
To further validate  the detection performance with multimodal data  in the compressed domain with  the proposed approaches, in this section,  we consider real experimental data. We use the footstep data, made available by the US Army
Research Laboratory (ARL), collected at the US southwest
border. The dataset consists of raw observations from several acoustic, seismic and PIR
sensors  that were deployed in an
outdoor space to record human and animal activity that is
typical in perimeter and border surveillance scenarios. The
participants in the data collection exercise walked/ran along
a predetermined path with sensors laid out along either side
of the path.

In the following experiments, we  consider two cases; detection of one man walking and a man leading a  horse  based on data at two sensors (one acoustic and one seismic). Each seismic/acoustic time series contains a
leading $60s$ of background data. For the detection problem, we use this as  $\mathcal H_0$ data.
The data are sampled at $10kHz$, and are mean centered and
oscillatory in nature. The time series   data at each sensor was split into non-overlapping frames of size $N$.  Further, $N_{tr}$ frames were used for training under each hypothesis  and $N_{mont} $ frames were used for test.

\begin{figure}[h!]
    \centering
    \begin{subfigure}[b]{0.38\textwidth}
        \includegraphics[width=\textwidth]{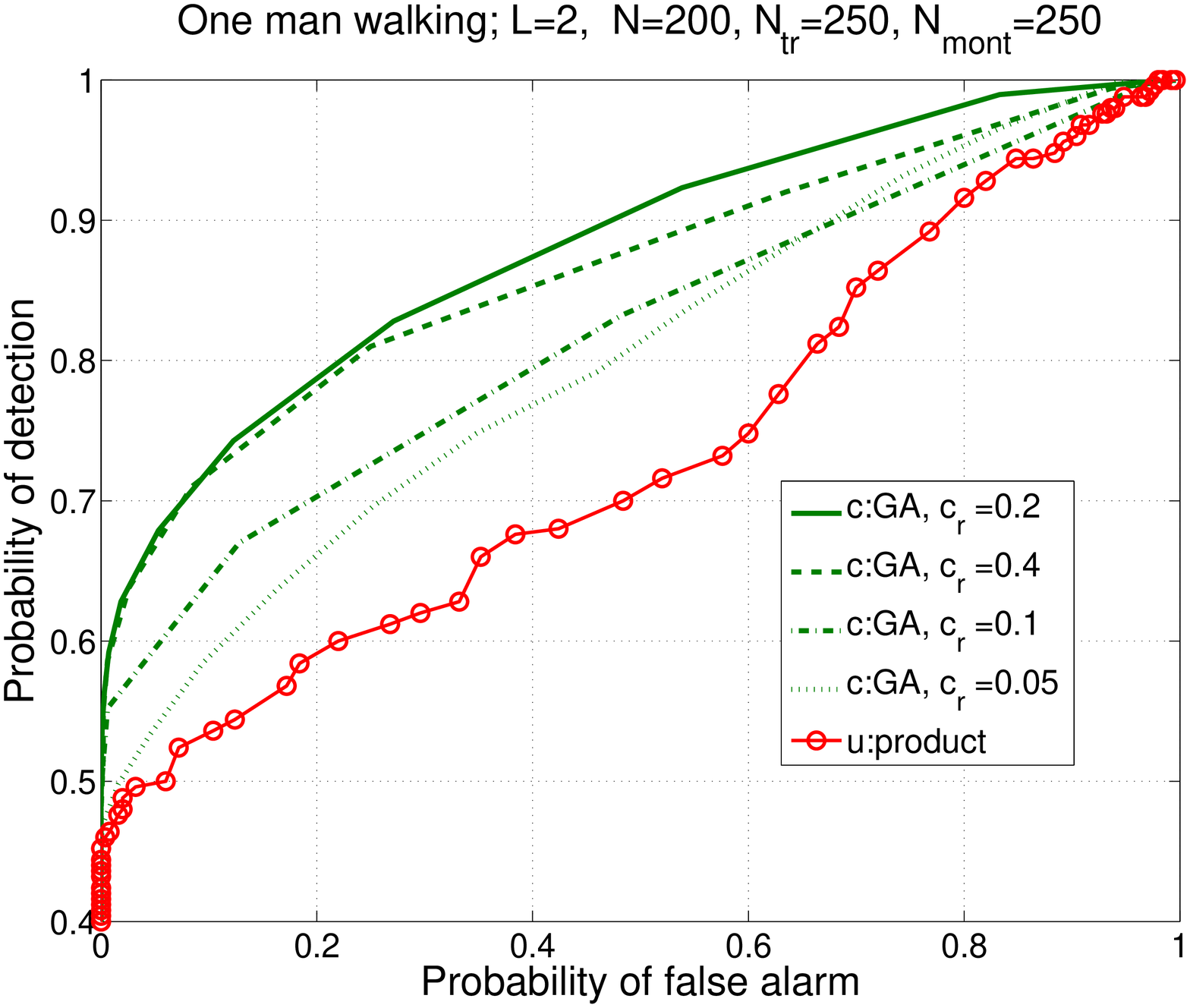}
        \caption{one man walking, $N=200$}
        \label{fig:N100}
    \end{subfigure}
    ~ 
    \begin{subfigure}[b]{0.38\textwidth}
        \includegraphics[width=\textwidth]{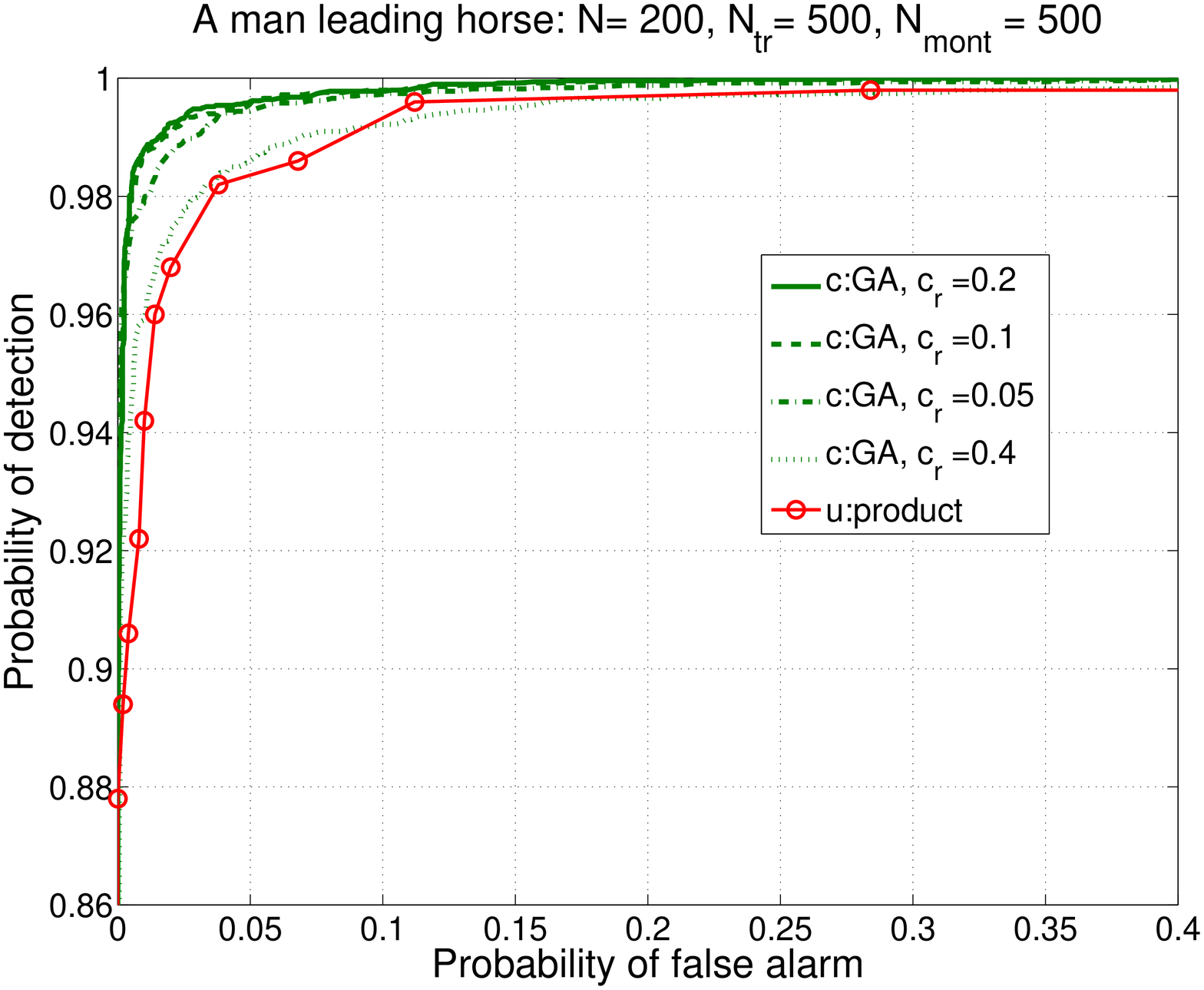}
        \caption{man leading horse; $N=200$}
        \label{fig:N1000}
    \end{subfigure}
    \caption{Detection performance with compressed and uncompressed data; $L=2$ (one seismic and one acoustic  sensor)}\label{fig:seismic_acoustic}
\end{figure}

In Fig. \ref{fig:seismic_acoustic}, we show the performance when detection is performed with `c:GA' (it is noted that we show the detection performance only with  the Gaussian approximation due to the limited number of samples to implement the covariance based approach). The mean and the covariance matrices of compressed data  under each hypotheses are  estimated using the training data. The values  used for $N$, $N_{tr}$ and $N_{mont}$ for the two cases are  shown in figure titles.  We further plot the detection performance with `u:product'.  To obtain the marginal pdfs of  uncompressed data under $\mathcal H_1$,  a kernel based density estimate  is computed  using the training data  with the Gaussian kernel. Under $\mathcal H_0$, the data is assumed to be iid Gaussian where the    mean and the  variance  were estimated using the training data.

  For both cases,  it is observed that `c:GA' with a small compression ratio, (e.g., $c_r = 0.05$),  outperforms  detection with `u:product'. Another  observation  is that, when $c_r$ increases beyond a certain threshold, the performance does not monotonically improve (e.g., performance with $c_r=0.2$ is better than that with $c_r=0.4$ in Fig. \ref{fig:seismic_acoustic}). This is because, when $c_r$  (thus $M$) increases, more training samples are needed  to estimate $\mathbf C^0$ and $\mathbf C^1$ accurately as  required in (\ref{stat_Gaussian}).  When the amount  of training data available is limited,    the estimates of   $\mathbf C^0$ and $\mathbf C^1$ become  less accurate as $M$ increases  leading to degraded detection performance. However, with the  available (limited) number of samples, detection with `c:GA' provides better performance  with small $c_r$ values than  detection using the product approach  with uncompressed data.

\section{Conclusion}\label{sec_conclusion}
Optimal decision fusion with high dimensional multimodal dependent data is a challenging problem. In this paper, we explored  the potential of CS in capturing the dependence structures of spatially dependent data  to develop  efficient decision  statistics  for detection in the compressed domain. In addition to the inherent benefits of CS in terms of  low computational and communication  overhead compared to processing and  transmitting  high dimensional data, we showed that the performance of  CS based detection with dependent data using LR can be better  than or similar to  several suboptimal detection techniques with uncompressed data under certain conditions.
 We discussed conditions under which   modeling dependence   in the compressed domain using Gaussian approximation  is more efficient and effective than modeling dependence with uncompressed data which is computationally expensive  most of the time. We further discussed a nonparametric approach for detection where a decision statistic is computed based on the covariance matrix of uncompressed data and the statistic is estimated in the compressed domain. This approach  can provide better performance when the non-Gaussian uncompressed  data is highly correlated with an additional computational cost compared to that is required for the  Gaussian approximation based approach. Further, the proposed compressed covariance based detector is more robust than the widely considered nonparametric detector; the energy detector.

\bibliographystyle{IEEEtran}
\bibliography{IEEEabrv,bib1,ref_1}

\end{document}